\documentclass[letterpaper]{article}
\usepackage[utf8]{inputenc}
\usepackage{geometry}
\usepackage{graphicx}
\usepackage{xcolor}
\usepackage[colorlinks = true,
            linkcolor = blue,
            urlcolor  = blue,
            citecolor = blue,
            anchorcolor = blue]{hyperref}
\usepackage{multirow}
\usepackage{hhline}
\usepackage{tikz}
\usepackage{pgfplots}

\pgfplotsset{width=7cm,compat=newest}

\usepackage{colortbl}
\usepackage{natbib}
\setcitestyle{numbers,open={[},close={]},comma}

\usepackage[font=small,labelfont=bf]{caption}

\usepackage{amsmath}
\usepackage{amsfonts}
\usepackage{amssymb}
\usepackage{amsthm}

\newtheorem{theorem}{Theorem}
\newtheorem{proposition}{Proposition}
\newtheorem{corollary}{Corollary}

\newcommand{\nash}[1]{{#1}^{\rm ne}}
\newcommand{\opt}[1]{{#1}^{\rm opt}}
\newcommand{\marg}[1]{{#1}^{\rm mc}}
\newcommand{\bb}[1]{\mathbb{#1}}
\newcommand{\mc}[1]{\mathcal{#1}}

\title{The Unintended Consequences of Minimizing the Price of Anarchy in Congestion Games}
\author{Rahul Chandan \and Dario Paccagnan \and Jason R. Marden}

\begin{document}

\maketitle

\begin{abstract}
This work focuses on the design of taxes in atomic congestion games, a commonly studied model for competitive resource sharing. While most related studies focus on optimizing either the worst- or best-case performance (i.e., Price of Anarchy (PoA) or Price of Stability (PoS)), we investigate whether optimizing for the PoA has consequences on the PoS. Perhaps surprisingly, our results reveal a fundamental trade-off between the two performance metrics. Our main result demonstrates that the taxation rule that optimizes the PoA inherits a matching PoS, implying that the best outcome is no better than the worst outcome under such a design choice.  We then study this trade-off in terms of the Pareto frontier between the PoA and PoS.  Our results also establish that any taxes with PoS equal to 1 incur a much higher PoA, and that, in several well-studied cases, the untaxed setting lies strictly above the Pareto frontier.
\end{abstract}

\section{Introduction}

We consider the design of systems that involve the interactions of strategic users and an underlying shared technological infrastructure.  A major difficulty in designing such systems is that one must account for each user's decision making process in order to guarantee good overall system performance.  The detrimental effects of selfish user behaviour on the performance of these systems have been observed in a variety of contexts, including unfair allocation of essential goods and services~\cite{cookson2016socio,dodorico2019food}, overexploitation of natural resources~\cite{hardin1968tragedy,kaida2016facilitating} and congestion in internet and road-traffic networks~\cite{cabannes2018impact,macfarlane2019your}.  A widely studied approach for influencing the system performance is the use of taxes, which can come in the form of rewards or penalties.  Examples of taxes include taxes levied on users whose decisions have a negative impact on the system performance, and rebates given to users for making decisions aligned with the greater good.

Many of the aforementioned systems can be suitably modeled as \emph{congestion games}~\cite{rosenthal1973class}, where users compete over a set of shared resources and the overall objective is to minimize the sum of users' costs.  In this model, users are assumed to be self interested, focused solely on minimizing their own experienced costs.  In order to measure the inefficiency of users' selfish decision making in such settings, several performance metrics have been proposed, including the \emph{Price of Anarchy}~\cite{koutsoupias1999worst} and the \emph{Price of Stability}~\cite{anshelevich2008price}.  Informally, the Price of Anarchy provides peformance guarantees for the worst system outcome, while the Price of Stability provides guarantees for the best system outcome.  A rich body of literature analyzes these two performance metrics in settings spanning far beyond the class of congestion games.  More recently, a number of works have focused on deriving taxes that optimize the Price of Anarchy in congestion games as a surrogate for optimizing the system performance~\cite{bilo2019dynamic,caragiannis2010taxes,paccagnan2019incentivizing}.  This prompts the following question:
\begin{list}{}{\leftmargin=1cm\rightmargin=1cm}
\item \emph{What are the consequences of optimizing for the Price of Anarchy on other performance metrics such as the Price of Stability?}
\end{list}
Perhaps surprisingly, we prove that there exists a fundamental trade-off between the Price of Anarchy and the Price of Stability in the class of congestion games.

\subsection{Model}

In this work, we consider the class of (atomic) congestion games as defined by Rosenthal~\cite{rosenthal1973class}.  A congestion game consists of a set of users $N=\{1,\dots,n\}$ and a set of resources $\mc{E}$.  Each user $i\in N$ must select an action $a_i$ from a corresponding set of feasible actions $\mc{A}_i\subseteq 2^\mc{E}$.  The cost that a user experiences for selecting a given resource $e\in\mc{E}$ depends only on the total number of users selecting $e$, and is denoted as $\ell_e:\{1,\dots,n\}\to\bb{R}$.  Given an assignment $a=(a_1,\dots,a_n) \in \mc{A}$, where $\mc{A} = \Pi_{i\in N}\mc{A}_i$, each user $i\in N$ experiences a cost equal to the sum over costs on resources $e\in a_i$.  Correspondingly, the system cost is measured by the sum of the users costs, i.e.,
\begin{equation} \label{eq:system_cost}
    {\rm SC}(a) = \sum_{i\in N} \sum_{e \in a_i} \ell_e(|a|_e)
\end{equation}
where $|a|_e$ denotes the number of users selecting resource $e$ in assignment $a$.  Observe that a congestion game can be represented as a tuple $G=(N,\mc{E},\{\mc{A}_i\}_{i\in N},\{\ell_e\}_{e\in\mc{E}})$.  We denote by $\mc{G}^{n,\mc{L}}$ the family of all congestion game instances with a maximum number of users $n$, where all resource cost functions $\{\ell_e\}_{e\in\mc{E}}$ belong to a common family of resource cost functions $\mc{L}$.  To ease the notation, we will use $\mc{G}^{n}$ to refer to the family $\mc{G}^{n,\mc{L}}$ when the dependence on $\mc{L}$ is clear.

\vspace{.2cm}\noindent\emph{Taxes.} In the study of taxes, each resource $e\in\mc{E}$ is associated with a tax function $\tau_e:\{1,\dots,n\}\to \bb{R}$ (positive or negative).  In this case, each user $i\in N$ incurs a cost involving both the resource costs it experiences and the imposed taxes, i.e., 
\begin{equation} \label{eq:user_cost_functions}
    C_i(a) = \sum_{e \in a_i} \Big[ \ell_e(|a|_e)+\tau_e(|a|_e) \Big].
\end{equation} 
We consider taxes that only influence the users' costs and do not factor into the social cost.  Scenarios where taxes are incorporated into the social cost have also been studied in, e.g., \cite{bilo2019dynamic,cole2006much}.  

When users selfishly choose their actions to minimize their incurred costs, an emergent outcome is often described by a \emph{pure Nash equilibrium}.  A pure Nash equilibrium is an assignment $\nash{a}\in\mc{A}$ such that $C_i(\nash{a}) \leq C_i(a_i, \nash{a}_{-i})$ for all $a_i\in\mc{A}_i$ and all $i\in N$, where $a'_i,a_{-i}$ denotes the assignment obtained when user $i$ plays action $a'_i$ and the remaining users continue to play their actions in $a$.  Observe that a system designer can influence the set of pure Nash equilibria through the choice of tax functions $\tau_e$, $e\in\mc{E}$.  Accordingly, with abuse of notation, we augment the tuple representation of a congestion game as $G=(N,\mc{E},\{\mc{A}_i\}_{i\in N},\{\ell_e\}_{e\in\mc{E}},\{\tau_e\}_{e\in\mc{E}})$ which incorporates the imposed taxes on the resources.

We consider the use of \emph{local} taxes to improve the equilibrium performance.  Local taxes only use information about the resource cost function $\ell_e$ to compute the tax function $\tau_e$ on any given resource $e \in \mc{E}$.  The restriction to local taxes is a natural requirement, especially in settings where scalable and computationally simple rules are desirable.  This structure is also commonly utilized in the existing literature, e.g., Pigouvian taxes~\cite{pigou1920economics}.  Accordingly, we define a local taxation rule as a map from the set of admissible resource costs $\mc{L}$ to taxes, i.e., under a given local taxation rule $T$, the tax function associated with each resource $e\in\mc{E}$ is given by $\tau_e = T(\ell_e)$.  For any given family of congestion games $\mc{G}$ without taxes, we denote by $\mc{G}_T$ the corresponding modified family of congestion games with taxes $\tau_e = T(\ell_e)$ on each edge $e\in\mc{E}$.

\vspace{.2cm}\noindent\emph{Performance metrics.}  For any given family of congestion game instances $\mc{G}$, we measure the equilibrium performance using two commonly-studied metrics termed \emph{Price of Anarchy} and \emph{Price of Stability}, respectively defined as
\begin{align}
    {\rm PoA}(\mc{G}) = \sup_{G\in\mc{G}} 
        \max_{a\in{\rm NE}(G)} \frac{{\rm SC}(a)}{{\rm MinCost}(G)}, \label{eq:poa}\\
    {\rm PoS}(\mc{G}) = \sup_{G\in\mc{G}} 
        \min_{a\in{\rm NE}(G)} \frac{{\rm SC}(a)}{{\rm MinCost}(G)}, \label{eq:pos}
\end{align} 
where ${\rm MinCost}(G)$ denotes the minimum achievable social cost for instance $G$ as defined in \eqref{eq:system_cost} and ${\rm NE}(G)$ denotes the set of all pure Nash equilibria in $G$.  It is important to note that the set ${\rm NE}(G)$ must be non-empty for any congestion game $G$, and, thus, that the Price of Anarchy and Price of Stability are well-defined.  This holds since congestion games are potential games \cite{rosenthal1973class}, i.e., there exists a function $\Phi:\mc{A}\to\bb{R}$ such that
\begin{equation} \label{eq:potential_function}
    \Phi(a)-\Phi(a_i,a_{-i}) = C_i(a) - C_i(a'_i, a_{-i}), \quad \forall a_i\in\mc{A}_i, \forall i\in N, \forall a\in\mc{A}.
\end{equation}
Specifically, congestion games admit the following potential function \cite{rosenthal1973class}:
\begin{equation} \label{eq:rosenthal_potential}
    \Phi(a)=\sum_{e\in \cup a_i} \sum^{|a|_e}_{k=1} [ \ell_e(k)+\tau_e(k) ].
\end{equation} 
Observe that the Price of Anarchy provides guarantees on the performance of \emph{any} equilibrium in the set of games while the Price of Stability offers performance guarantees for the \emph{best} equilibrium of any instance in the set.  By definition, ${\rm PoA}(\mc{G}) \geq {\rm PoS}(\mc{G}) \geq 1$ for any family $\mc{G}$.  While we introduce the Price of Anarchy and Price of Stability with respect to pure Nash equilibria, we note that all of our results extend to coarse-correlated equilibria \cite{young2004strategic,roughgarden2015intrinsic}.  Therefore, in the remainder of this work, we use the Price of Anarchy and Price of Stability to refer to the efficiency of pure Nash and coarse-correlated equilibria, equivalently.

\subsection{Summary of our contributions}

In this work, we demonstrate that there exists an inherent trade-off between the Price of Anarchy and the Price of Stability in congestion games, and put forward techniques to study this trade-off.  A discussion of our results is provided below:

\vspace{.2cm}\noindent\emph{Optimizing for anarchy or stability.} Our first contribution characterizes the Price of Stability of the family of congestion games under taxes that minimize the Price of Anarchy.  Specifically, Theorem~\ref{thm:optimal_poa} establishes that taxes achieving minimum Price of Anarchy, ${\rm MinPoA}(n,\mc{L}) = \min_T {\rm PoA}(\mc{G}^{n,\mc{L}}_T)$, have corresponding Price of Stability \emph{equal} to the Price of Anarchy.  The result holds for any family of congestion games with convex, nondecreasing resource costs, including polynomial congestion games.\footnote{Convex, nondecreasing resource cost functions appear in many important problem settings, including congestion in transportation networks \cite{united1964traffic}, diseconomies of scale \cite{pigou1920economics,coase1960problem}, and multiagent coordination \cite{murphey2000target}.}  This result is perhaps unexpected as -- at least in the setting without taxes -- there can exist a significant gap between the Price of Anarchy and Price of Stability. For example, in polynomial congestion games, the Price of Anarchy grows exponentially in the order $d$ \cite{aland2011exact} while the Price of Stability grows only polynomially in $d$ \cite{christodoulou2015price}.  The result in Theorem~\ref{thm:optimal_poa} also implies -- by contrapositive -- that any improvement in the Price of Stability necessarily comes at the expense of the Price of Anarchy, as any rule $T$ with Price of Stability strictly less than ${\rm MinPoA}(n,\mc{L})$ necessarily has Price of Anarchy strictly greater than ${\rm MinPoA}(n,\mc{L})$.

As we review in the forthcoming Section \ref{sec:unilateral_pos}, the Price of Stability of any family of congestion games under marginal cost taxes (i.e., Pigouvian taxation \cite{pigou1920economics}) is equal to 1.  However, in order to establish that a trade-off exists in general, it remains to be shown whether there is a strict separation between the extreme points of the Pareto frontier.  In particular, note that if there is a tax that guarantees a Price of Anarchy equal to 1 for a given family of congestion games, then there can be no trade-off.  Following this line of reasoning, there \emph{must} exist a trade-off between Price of Anarchy and Price of Stability in all sets of congestion games with convex, nondecreasing latency functions (except constant), as it was recently shown by \citet{paccagnan2021congestion} that the best achievable Price of Anarchy among all taxes (local or otherwise) is strictly greater than 1.  Another consequence of combining Theorem~\ref{thm:optimal_poa} with the result in \cite{paccagnan2021congestion} is that the marginal cost rule -- the unique rule with Price of Stability of 1 (Proposition \ref{prop:marginal_unique}) -- can never provide optimal Price of Anarchy.

\vspace{.2cm}\noindent\emph{The anarchy-stability trade-off in congestion games.} Our second contribution goes beyond the commonly studied cases discussed above, characterizing upper and lower bounds on the Pareto frontier corresponding with the Price of Anarchy, Price of Stability trade-off. For a given family of resource cost functions $\mc{L}$ and maximum number of users $n$, we show that any local rule $T$ satisfying a Price of Anarchy requirement ${\rm PoA}(\mc{G}^n_T)\leq\bar{\Pi}$ -- with $\bar{\Pi}\geq{\rm MinPoA}(n,\mc{L})$ -- cannot achieve a Price of Stability below a lower bound ${\rm LB}(\bar{\Pi})$ that is nonincreasing in $\bar{\Pi}$.  Furthermore, the best achievable Price of Stability among all such rules does not exceed an upper bound ${\rm UB}(\bar{\Pi})$, also nonincreasing in $\bar{\Pi}$.  Together, ${\rm UB}(\bar{\Pi})$ and ${\rm LB}(\bar{\Pi})$ -- respectively derived in Theorems \ref{thm:upperbound} and \ref{thm:lower_bound} -- provide a characterization of the Pareto frontier.  Specifically, any joint performance guarantee falling strictly below the lower bound curve is unachievable by any local rule, while those strictly above the upper bound curve are suboptimal, i.e., there exist local taxes providing strict improvement in the Price of Anarchy while guaranteeing at most the same Price of Stability or vice versa.

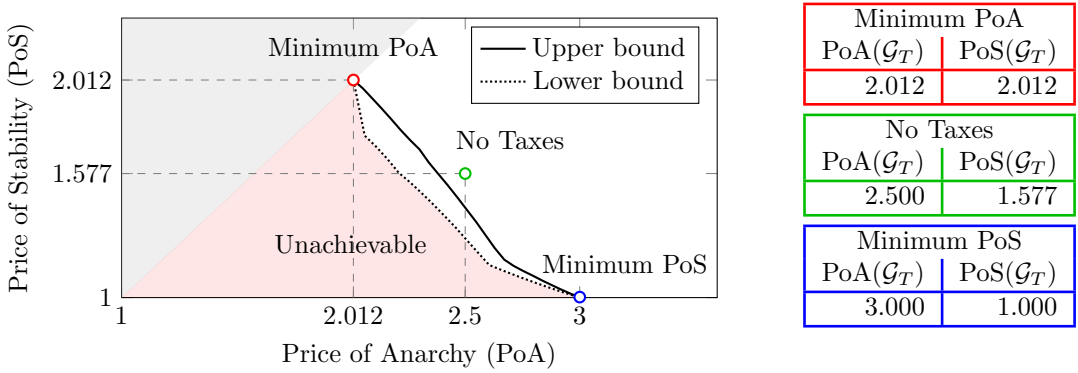
\begin{figure}[b!]
    \centering
    \caption{\emph{The Pareto frontier between the Price of Anarchy and Price of Stability in congestion games with affine and quadratic resource costs.} The Pareto frontier lies within the region below the upper bound curves (solid black) and above the lower bound curves (dotted black), which were derived with the techniques put forward in Theorems \ref{thm:upperbound} and \ref{thm:lower_bound}. The joint Price of Anarchy and Price of Stability values for the mechanism that minimizes the Price of Anarchy (in red), no taxes (in green), and the mechanism that minimizes the Price of Stability (in blue) are reported in the table on the right.  Note that all joint performance guarantees in the region above the upper bound curves are suboptimal, in the pink region below the lower bound curves are unachievable by any local mechanism and in the grey region are inadmissible as ${\rm PoA}(\mc{G}_T)\geq{\rm PoS}(\mc{G}_T)$ must hold.  Although the upper and lower bound curves do not always match, we show that they are tight at the endpoints for any family of convex, nondecreasing resource costs.  Similar bounds on the Pareto frontier can be derived for any family of resource cost functions using the techniques outlined in Section~\ref{sec:bounds}.}
    \label{fig:pospoa}    
    \begin{minipage}{0.7\textwidth}
    \begin{tikzpicture}
        \begin{axis}[
            width=9.5cm,
            height=5.3cm,
            xlabel={Price of Anarchy (${\rm PoA}$)},
            ylabel={Price of Stability (${\rm PoS}$)},
            xtick = {1,2.012,2.5,3},
            xticklabels= {1,2.012,2.5,3},
            ytick = {1,1.577,2.012},
            yticklabels = {1,1.577,2.012},
            legend pos=north east,
            xmin = 1.0,
            xmax = 3.6,
            ymin=1.0, 
            ymax = 2.3,
        ]
            \coordinate (r1) at (2.012,1);
            \coordinate (r2) at (2.5,1);
            \coordinate (r3) at (3,1);
            \coordinate (g1) at (1,2.012);
            \coordinate (g2) at (1,1.577);
            \coordinate (o) at (1,1);
            \coordinate (d) at (3,3);
            \coordinate (s) at (1,3);
            \coordinate (C) at (2.5,1.577);
            \coordinate (minPoA) at (2.012,2.012);
            \filldraw[white!50!gray,opacity=0.25] (o) -- (d) -- (s) -- cycle;
            \filldraw[red,draw=none,opacity=0.1] (o) -- (minPoA) -- (r1) -- cycle;
            \addplot+[no marks,draw=none,fill=red,fill opacity=0.1,forget plot] table {lpos.dat} \closedcycle;
            \addplot+ [no marks,thick,black] table {aupos.dat}
                coordinate [pos=0.0] (A)
                coordinate [pos=1.0] (B);
            \addlegendentry{Upper bound};
            \addplot+ [no marks,thick,densely dotted,black] table {lpos.dat};
            \addlegendentry{Lower bound};
            \addplot+ [only marks,thick,mark options={fill=white},green!75!black,mark=*] coordinates {(2.5,1.577)};
            \draw[dashed,gray] (r1) -- (A);
            \draw[dashed,gray] (g1) -- (A);
            \draw[dashed,gray] (r2) -- (C);
            \draw[dashed,gray] (g2) -- (C);
            \node () at (2,1.25)  {Unachievable};
        \end{axis}
        \draw[red,thick,fill=white]  (A) circle (2pt);
        \node () at (A) [label={[xshift=0.cm, yshift=0.075cm]Minimum ${\rm PoA}$}] {};
        \draw[blue,thick,fill=white]  (B) circle (2pt);
        \node () at (B) [label={[xshift=0.6cm, yshift=0.075cm]Minimum ${\rm PoS}$}] {};
        \node () at (C) [label={[xshift=0.6cm, yshift=0.075cm]No Taxes}] {};
    \end{tikzpicture}
    \end{minipage}
    \begin{minipage}{0.25\textwidth}
    \vspace*{-23pt}
    {\global\arrayrulewidth=1pt}
    \arrayrulecolor{red}
    \begin{tabular}{|r|r|}
        \hline
        \multicolumn{2}{|c|}{Minimum ${\rm PoA}$} \\
        ${\rm PoA}(\mc{G}_T)$ & ${\rm PoS}(\mc{G}_T)$ \\ \hline
        2.012 & 2.012 \\ \hline
    \end{tabular}

    \vspace{2pt}
    \arrayrulecolor{green!75!black}
    \begin{tabular}{|r|r|}
        \hline 
        \multicolumn{2}{|c|}{No Taxes} \\
        ${\rm PoA}(\mc{G}_T)$ & ${\rm PoS}(\mc{G}_T)$ \\ \hline
        2.500 & 1.577 \\ 
        \hline 
    \end{tabular}

    \vspace{2pt}
    \arrayrulecolor{blue}
    \begin{tabular}{|r|r|}
        \hline 
        \multicolumn{2}{|c|}{Minimum ${\rm PoS}$} \\
        ${\rm PoA}(\mc{G}_T)$ & ${\rm PoS}(\mc{G}_T)$ \\ \hline
        3.000 & 1.000 \\ 
        \hline 
    \end{tabular}
    \end{minipage} \\ \vspace{.1cm}(a)~Congestion games with affine resource costs. \\\vspace{.5cm}
    \begin{minipage}{0.7\textwidth}
    \begin{tikzpicture}
        \begin{axis}[
            width=9.5cm,
            height=5.3cm,
            xlabel={Price of Anarchy (${\rm PoA}$)},
            ylabel={Price of Stability (${\rm PoS}$)},
            xtick = {1,5.101,9.583,13},
            xticklabels= {1,5.101,9.583,13},
            ytick = {1,2.361,5.101},
            yticklabels = {1,2.361,5.101},
            legend pos=north east,
            xmin = 1.0,
            xmax = 14.0,
            ymin=1.0, 
            ymax = 6.2,
        ]
            \coordinate (r1) at (5.101,1);
            \coordinate (r2) at (9.5833,1);
            \coordinate (r3) at (13,1);
            \coordinate (g1) at (1,5.101);
            \coordinate (g2) at (1,2.381);
            \coordinate (o) at (1,1);
            \coordinate (d) at (20,20);
            \coordinate (s) at (1,20);
            \coordinate (C) at (9.5833,2.361);
            \coordinate (minPoA) at (5.101,5.101);
            \filldraw[white!50!gray,opacity=0.25] (o) -- (d) -- (s) -- cycle;
            \filldraw[red,draw=none,opacity=0.1] (o) -- (minPoA) -- (r1) -- cycle;
            \addplot+[no marks,draw=none,fill=red,fill opacity=0.1,forget plot] table {lpos2.dat} \closedcycle;
            \addplot+ [no marks,thick,black] table {aupos2.dat}
                coordinate [pos=0.0] (A)
                coordinate [pos=1.0] (B);
            \addlegendentry{Upper bound};
            \addplot+ [no marks,thick,densely dotted,black] table {lpos2.dat};
            \addlegendentry{Lower bound};
            \addplot+ [only marks,thick,mark options={fill=white},green!75!black,mark=*] coordinates {(9.583,2.361)};
            \draw[dashed,gray] (r1) -- (A);
            \draw[dashed,gray] (g1) -- (A);
            \draw[dashed,gray] (r2) -- (C);
            \draw[dashed,gray] (g2) -- (C);
            \node () at (4.5,1.75)  {Unachievable};
        \end{axis}
        \draw[red,thick,fill=white]  (A) circle (2pt);
        \node () at (A) [label={[xshift=0.cm, yshift=0.075cm]Minimum ${\rm PoA}$}] {};
        \draw[blue,thick,fill=white]  (B) circle (2pt);
        \node () at (B) [label={[xshift=-0.67cm, yshift=0.075cm]Minimum ${\rm PoS}$}] {};
        \node () at (C) [label={[xshift=0.cm, yshift=0.075cm]No Taxes}] {};
    \end{tikzpicture}
    \end{minipage}
    \begin{minipage}{0.25\textwidth}
    \vspace*{-24pt}
    {\global\arrayrulewidth=1pt}
    \arrayrulecolor{red}
    \begin{tabular}{|r|r|}
        \hline
        \multicolumn{2}{|c|}{Minimum ${\rm PoA}$} \\
        ${\rm PoA}(\mc{G}_T)$ & ${\rm PoS}(\mc{G}_T)$ \\ \hline
        5.101 & 5.101 \\ \hline
    \end{tabular}

    \vspace{2pt}
    \arrayrulecolor{green!75!black}
    \begin{tabular}{|r|r|}
        \hline 
        \multicolumn{2}{|c|}{No Taxes} \\
        ${\rm PoA}(\mc{G}_T)$ & ${\rm PoS}(\mc{G}_T)$ \\ \hline
        9.583 & 2.361 \\ 
        \hline 
    \end{tabular}

    \vspace{2pt}
    \arrayrulecolor{blue}
    \begin{tabular}{|r|r|}
        \hline 
        \multicolumn{2}{|c|}{Minimum ${\rm PoS}$} \\
        ${\rm PoA}(\mc{G}_T)$ & ${\rm PoS}(\mc{G}_T)$ \\ \hline
        13.000 & 1.000 \\ 
        \hline 
    \end{tabular}
    \end{minipage} \\ \vspace{.1cm}(b)~Congestion games with quadratic resource costs.
\end{figure}

{\global\arrayrulewidth=0.4pt}
\arrayrulecolor{black}

By virtue of the techniques used to derive ${\rm UB}(\bar{\Pi})$ and ${\rm LB}(\bar{\Pi})$, three important observations can be made.  First, for any family of convex, nondecreasing resource costs, the upper and lower bound match at extreme values of $\bar{\Pi}$, i.e., ${\rm LB}(\bar{\Pi})={\rm UB}(\bar{\Pi})$, when $\bar{\Pi}={\rm MinPoA}(n,\mc{L})$ or ${\rm PoS}(\mc{G}^n_T)=1$.  Second, we are able to establish that the joint performance guarantees without taxation are strictly suboptimal for several well-studied families of resource cost functions (Section \ref{sec:efficiency_notax}).  Third, the bounds ${\rm UB}(\bar{\Pi})$ and ${\rm LB}(\bar{\Pi})$ extend to general notions of equilibrium (e.g., coarse-correlated equilibria) and to other performance metrics.\footnote{Specifically, the bounds also apply to the performance of potential minimizers of the game, i.e., assignments satisfying $a^{\rm pm}\in\arg\min_{a\in\mc{A}}\Phi(a)$. The worst-case ratio between ${\rm SC}(a^{\rm pm})$ and ${\rm MinCost}(G)$ across a family of games has previously been studied under the terminology of \emph{Price of Stochastic Anarchy} in, e.g., \cite{chung2008price}.  The Price of Stochastic Anarchy is particularly relevant when considering the performance of noisy learning dynamics that converge to a potential minimizer with probability 1 \cite{alos2010logit,blume1993statistical,marden2012revisiting,young2001individual}.}

In Figure~\ref{fig:pospoa}, we plot the upper bound (solid, black lines) and lower bound curves (dotted, black lines) on the Pareto frontier between the Price of Anarchy and Price of Stability in affine and quadratic congestion games.  We also provide the Price of Anarchy and Price of Stability of the rules that minimize the Price of Anarchy (red circle), the rules that minimize the Price of Stability (blue circle) and no tax (green circle) in the tables on the right.  Under local rules, the best achievable Price of Anarchy ${\rm MinPoA}(n,\mc{L})$ is approximately 2.012 for affine resource costs and 5.101 for quadratic resource costs \cite{paccagnan2019incentivizing}, and -- in accordance with Theorem~\ref{thm:optimal_poa} -- the corresponding Price of Stability values are also 2.012 and 5.101, respectively.

From the tables in Figure~\ref{fig:pospoa}, observe that utilizing either no tax or the rule that minimizes the Price of Stability yields Price of Stability strictly lower than the best achievable Price of Anarchy ${\rm MinPoA}(n,\mc{L})$.  However, they also have strictly higher Price of Anarchy than ${\rm MinPoA}(n,\mc{L})$.  Finally, observe that neither of the rules that unilaterally minimize for either performance metric (red and blue circles) is strictly ``better'' (i.e., Pareto dominant) than the no tax setting (green circle) as they do not guarantee that \emph{both} the Price of Anarchy and Price of Stability of the modified family of games are reduced.  Nevertheless, the joint Price of Anarchy and Price of Stability of no tax falls above the upper bound curves in both plots, implying that Pareto dominant rules do exist in these settings.  Finally, note that joint Price of Anarchy and Price of Stability values below the lower bound curves cannot be achieved by \emph{any} local rule.

\vspace{.2cm}\noindent\emph{Attainable joint performance guarantees.} Our results thus far consider the Price of Anarchy and Price of Stability as independent performance metrics, i.e., as independent worst-case measures of the equilibrium efficiency across the games in a given family of instances.  In particular, we remark that the Price of Anarchy and Price of Stability of a family of games (as defined in \eqref{eq:poa} and \eqref{eq:pos}) are not necessarily attained within the same game instance, and may correspond to two separate game instances in general.  In our final set of contributions, we seek to explore attainable joint performance guarantees by considering only those Price of Anarchy and Price of Stability pairs that are achieved within the same game instance.  For a given local taxation rule, we establish that the attainable joint performance guarantees do not always match those obtained from the independent performance metrics we have considered thus far (Corollary \ref{cor:gamebygame_ub}).  Nonetheless, we show that a trade-off between the Price of Anarchy and Price of Stability persists under the restriction to attainable performance guarantees, and that -- perhaps surprisingly -- the extreme points of the corresponding trade-off curve coincide with those under the independent joint performance measure (Corollary \ref{cor:instance-by-instance}). Specifically, the same game instance \emph{does} achieve the independent joint performance guarantee under taxes that minimize either the Price of Anarchy or Price of Stability.

\subsection{Related works}

The Price of Anarchy was introduced by \citet{koutsoupias1999worst} as a performance metric to characterize the equilibrium efficiency in games.  The first exact characterization of the Price of Anarchy in congestion games was derived independently by \citet{awerbuch2013price} and \citet{christodoulou2005price} for affine congestion games without taxes.  These results were later generalized to all polynomial congestion games without taxes by \citet{aland2011exact}.

Characterizations of the Price of Anarchy without taxes naturally led to the study of taxes to improve worst-case efficiency guarantees.  The design of taxes to optimize equilibrium efficiency guarantees falls under the broader literature on coordination rules introduced by \citet{christodoulou2004coordination}.  Within the context of congestion games, \citet{caragiannis2010taxes} derive local and global congestion-independent rules that minimize the Price of Anarchy for linear resource costs.  For polynomial congestion games, \citet{bilo2019dynamic} consider the class of rules that use only information about a social optimum of each instance.  Among taxation rules of this specialized class, they derive the best achievable Price of Anarchy guarantees in polynomial congestion games, as well as a methodology for computing the optimal taxes.  \citet{paccagnan2021congestion} generalize these results beyond polynomial congestion games and show that the efficiency of optimal, polynomially-computable taxes (using \emph{global} information) matches the corresponding bound on the hardness of approximation.  \citet{paccagnan2019incentivizing} derive \emph{local} taxes that minimize the Price of Anarchy in any class of congestion games, which are shown to have similar efficiency guarantees as the rules using global information from \cite{bilo2019dynamic,caragiannis2010taxes,paccagnan2021congestion}.  \citet{bjelde2017brief} derive upper bounds on the Price of Anarchy associated with the marginal cost rule in polynomial congestion games, which were later refined and generalized in \citep{paccagnan2019incentivizing}.

Aside from the Price of Anarchy, another interesting metric that has been the subject of extensive analysis is the Price of Stability.  The Price of Stability was defined by \citet{anshelevich2008price} (though its study dates back even earlier to, e.g., \citet{schulz2003performance}), who provide an exact characterization of this metric for a specialized class of congestion games.  The exact Price of Stability for linear congestion games without taxes was derived by \citet{caragiannis2011tight} and~\citet{christodoulou2005price2}, followed by an exact characterization for all polynomial congestion games without taxes by \citet{christodoulou2015price}.  \citet{kleer2019tight} study the Price of Anarchy and Price of Stability in affine congestion games under various taxation rules (e.g., altruism, congestion independent taxes).

Though these and many other works on congestion games study the Price of Anarchy and the Price of Stability independently, there is no systematic framework for analyzing the concurrent optimization of these two performance metrics.  For example, one may wish to identify the taxes that optimize the Price of Stability while guaranteeing that the Price of Anarchy remains below some maximum acceptable value.  The insights developed in this paper directly relate to the design of optimal taxes in congestion games, accounting for possible trade-offs between the worst and best case performance metrics.  A natural concern in the joint study of the Price of Anarchy and Price of Stability under taxes is whether optimizing for the worst-case equilibrium's efficiency has negative consequences on the best-case equilibrium's efficiency.  

To the best of our knowledge, only two prior works have proposed the study of trade-offs between the Price of Anarchy and Price of Stability, albeit in specialized settings.  \citet{filos2019pareto} study the trade-off between the Price of Anarchy and the Price of Stability for unrelated machine scheduling.  Before that, the existence of a trade-off between these two performance metrics was also discovered by \citet{ramaswamy2017impact} for the class of covering games \cite{gairing2009covering}.  These two prior works both answer the above concern in the affirmative, providing characterizations of the Pareto frontier between the Price of Anarchy and Price of Stability within their respective problem settings.  Interestingly, though they consider specialized classes of problems, the two works show that the rules/mechanisms that optimize the Price of Anarchy have Price of Anarchy equal to the Price of Stability, which is mirrored by the result in Theorem \ref{thm:optimal_poa}.  In this manuscript, we investigate the same research direction as these two prior studies in the much broader class of congestion games, and show that their findings hold more generally.  We also extend the study of the Price of Anarchy and Price of Stability to attainable joint performance guarantees.

\section{Preliminaries} \label{sec:prelims}

In the literature on congestion games, the resource cost functions are often taken to be polynomials with nonnegative coefficients \cite{aland2011exact,bilo2019dynamic,caragiannis2010taxes,christodoulou2005price,christodoulou2015price,paccagnan2019incentivizing}, i.e., 
\[ \ell_e(x) = \alpha_1 + \alpha_2 x + \dots + \alpha_{d+1} x^d, \quad \forall e\in \mc{E}, \]
where $\alpha_j\geq 0$, $j=1,\dots,d+1$.  Note that one could equivalently consider the family of congestion games with resource cost functions $\ell_e\in \mc{L}$ where $\mc{L}$ is the set of all linear combinations with nonnegative coefficients, $\ell(x)=\sum^{d+1}_{j=1} \alpha_j b_j(x)$, of the set of polynomial basis functions $b_j(x) = x^{j-1}$, $j=1,\dots,d+1$.  Throughout this paper, we consider a broader family of congestion games corresponding to the family of resource cost functions $\mc{L}=\text{span}(b_1,\dots,b_m)$ containing all linear combinations with nonnegative coefficients of a set of basis functions $b_1,\dots,b_m$.

\subsection{Minimizing the Price of Anarchy}
\citet{paccagnan2019incentivizing} provide a linear programming based methodology to compute the local taxation rule that minimizes the Price of Anarchy in a given congestion game.  For the reader's convenience, we reproduce this linear programming methodology in the following proposition:

\begin{proposition}[Theorem~2.1~\citep{paccagnan2019incentivizing}] \label{prop:optimal_poa}
Consider the family of resource cost functions $\mc{L}=\text{span}(b_1,\dots,b_m)$ corresponding to basis functions $b_1, \dots, b_m$, and maximum number of users $n$.  Define 
\begin{equation} \label{eq:In}
    \mc{I}(n) := \left\{ (x,y,z) \in \{0,1,\dots,n\}^3 \text{ s.t. } 1\leq x+y-z\leq n \text{ and } z \leq \min\{x,y\} \right\},
\end{equation}
and let $\rho^{\rm PoA}_j\in\bb{R}$, $F^{\rm PoA}_j:\{1,\dots,n\}\to\bb{R}$ solve the following $m$ linear programs (one for each $b_j$):\footnote{Note that further reductions to the set $\mc{I}(n)$ are possible (see, e.g., \cite{paccagnan2019incentivizing}), but are omitted for conciseness.}
\begin{equation} \label{linprog:optimalpoa}
\begin{aligned}
    \underset{F,\rho}{\text{maximize}} \quad & \rho \\
    \text{subject to:} \quad & b_j(y)y-\rho b_j(x)x + F(x)(x-z) - F(x+1)(y-z) \geq 0,
        \quad \forall (x,y,z) \in \mc{I}(n),
\end{aligned}
\end{equation}
where we define $b_j(0)=F(0)=F(n+1)=0$.  It holds that a local taxation rule $T^{\rm PoA} \in \arg\min_T {\rm PoA}(\mc{G}^n_T)$ is given by
\begin{equation} \label{def:minimize_poa}
    T^{\rm PoA}(\ell) = \sum^m_{j=1} \alpha_j \tau^{\rm PoA}_j, 
        \quad \text{where } \tau^{\rm PoA}_j:\{1,\dots,n\}\to\bb{R},
        \quad \tau^{\rm PoA}_j(x)=F^{\rm PoA}_j(x)-b_j(x),
\end{equation}
and the corresponding Price of Anarchy is ${\rm PoA}(\mc{G}^n_{T^{\rm PoA}})=\max_j \{1/\rho^{\rm PoA}_j\} =: {\rm MinPoA}(n,\mc{L})$.
\end{proposition}

\noindent The constraints in \eqref{linprog:optimalpoa} -- and in forthcoming convex programs -- stem from efficient game parameterizations and smoothness conditions that allow the problem of deriving a bound on the Price of Anarchy or Price of Stability of a family of games to be recast as a tractable optimization problem (see, e.g., \cite{bilo2018unifying,chandan2020optimal,nadav2010limits}).  Under the parameterization, each constraint corresponds to a specific game structure within the family of games.  In some settings, including that of \eqref{linprog:optimalpoa}, such bounds can be shown to be tight for broad classes of games.

\subsection{Minimizing the Price of Stability} \label{sec:unilateral_pos}
Along with the local taxation rules that minimize the Price of Anarchy, we are naturally also interested in local taxes that minimize the Price of Stability. In the next proposition, we provide an intuitive argument to show that the marginal cost rule is the unique local taxation rule that minimizes the Price of Stability.

\begin{proposition} \label{prop:marginal_unique}
Consider the family of resource cost functions $\mc{L}=\text{span}(b_1,\dots,b_m)$ corresponding to positive, nondecreasing basis functions $b_1, \dots, b_m$, and maximum number of users $n$.  Then, the marginal cost rule
\[   T^{\rm mc}(\ell) = \sum^m_{j=1} \alpha_j \tau^{\rm mc}_j, 
        \quad \text{where } \tau^{\rm mc}_j:\{1,\dots,n\}\to\bb{R},
        \quad \tau^{\rm mc}_j(x)=(x-1)[b_j(x)-b_j(x-1)] \]
is the unique local taxation rule (up to rescaling) with ${\rm PoS}(\mc{G}^n_{\marg{T}})=1$.
\end{proposition}

\begin{proof}
We prove the claim in two parts: (i) show that any local taxation rule $T\neq T^{\rm mc}$ has ${\rm PoS}(\mc{G}^n_T)>1$; and (ii) prove that an optimal assignment is an equilibrium under $T^{\rm mc}$.

\emph{Part~(i):} Assume, by contradiction, that there exists a local taxation rule $T$ with ${\rm PoS}(\mc{G}^n_T)=1$ with $T(\ell)(k)>T^{\rm mc}(\ell)(k)$ for some integer $1\leq k \leq n$ and resource cost function $\ell\in\mc{L}$.  Consider the game $G$ with user set $N=\{1,\dots,k\}$ and two resources $\mc{E} = \{e_0, e_1\}$.  The resource $e_0$ has resource cost $\ell(x)$ and resource $e_1$ has resource cost $[\ell(k)k-\ell(k-1)(k-1)+\epsilon]\cdot\ell(x)$ for $x=1,\dots,k$ where $0<\epsilon<T(\ell)(k)-T^{\rm mc}(\ell)(k)$.  Every user $i\in\{1,\dots,k-1\}$ has only one action, $a_i=\{e_0\}$, while user $k$ has action set $\mc{A}_k=\{a_k,a'_k\}$, where $a_k=\{e_0\}$ and $a'_k=\{e_1\}$.  Observe that if $T(\ell)(k)>T^{\rm mc}(\ell)(k)$, then the unique pure Nash equilibrium corresponds with when user $k$ selects $a'_k$ resulting in social cost $\ell(k-1)(k-1)+\ell(k)k-\ell(k-1)(k-1)+\epsilon$.  Thus, the Price of Stability in this game is $[\ell(k)k+\epsilon]/[\ell(k)k]>1$, which contradicts ${\rm PoS}(\mc{G}^n_T)=1$.  We conclude this part by observing that a similar argument holds for $T(\ell)(k)<T^{\rm mc}(\ell)(k)$, when the resource cost of $e_1$ is $[\ell(k)k-\ell(k-1)(k-1)-\epsilon]\cdot\ell(x)$ for $x=1,\dots,k$ and $0<\epsilon<T^{\rm mc}(\ell)(k)-T(\ell)(k)$. In this case, user $k$'s Nash action is $a_k$ and the Price of Stability is $\ell(k)k/[\ell(k)k-\epsilon]>1$.

\emph{Part (ii):} Consider an optimal assignment $\opt{a}$ in a given game $G$.  It is straightforward to show that this assignment must be an equilibrium under $T^{\rm mc}$:
\begin{align*}
    C_i(\opt{a})-C_i(a_i,\opt{a}_{-i})
    =\> & \sum_{e\in\opt{a}_i} \Big[\ell_e(|\opt{a}|_e)|\opt{a}|_e-\ell_e(|\opt{a}|_e-1)(|\opt{a}|_e-1)\Big] \\
        & \qquad - \sum_{e\in a_i}\Big[\ell_e(|a_i,\opt{a}_{-i}|_e)|a_i,\opt{a}_{-i}|_e-\ell_e(|a_i,\opt{a}_{-i}|_e-1)(|a_i,\opt{a}_{-i}|_e-1)\Big] \\
    =\> & \sum_{e\in\opt{a}_i\setminus a_i}\ell_e(|\opt{a}|_e)|\opt{a}|_e + \sum_{e\in a_i\setminus\opt{a}_i} \ell_e(|a_i,\opt{a}_{-i}|_e-1)(|a_i,\opt{a}_{-i}|_e-1) \\
        & \qquad - \sum_{e\in\opt{a}_i\setminus a_i} \ell_e(|\opt{a}|_e-1)(|\opt{a}|_e-1) - \sum_{e\in a_i\setminus\opt{a}_i} \ell_e(|a_i,\opt{a}_{-i}|_e)|a_i,\opt{a}_{-i}|_e] \\
    =\> & \sum_{e\in\mc{E}}\ell_e(|\opt{a}|_e)|\opt{a}|_e-\sum_{e\in\mc{E}}\ell_e(|a_i,\opt{a}_{-i}|_e)|a_i,\opt{a}_{-i}|_e \\
    =\> & {\rm MinCost}(G)-{\rm SC}(a_i,\opt{a}_{-i}),
\end{align*}
where the third equality holds because we add and subtract $\ell_e(|\opt{a}|_e)|\opt{a}|_e$ for all $e\in\mc{E}\setminus(\opt{a}_i\cup a_i)$ and all $e\in\opt{a}_i\cap a_i$.  The final line must be nonpositive for all actions $a_i\in\mc{A}_i$ and all users $i\in N$ by the definition of ${\rm MinCost}(G)$, concluding the proof.
\end{proof}

In the above discussion, we identified the local taxes that unilaterally minimize the Price of Anarchy and the Price of Stability. In the next section, we consider the consequences of the unilateral minimization of either of these performance metrics on the value of the other.

\section{Optimizing for anarchy or stability} \label{sec:unilateral}

We first seek to investigate how the performance of the best case equilibria is affected when we optimize for the worst case equilibrium performance.  In the next result, we prove that any local taxation rule $T$ that minimizes the Price of Anarchy has corresponding Price of Stability equal to the Price of Anarchy in congestion games with convex, nondecreasing resource cost functions.  Additionally, we show that linear taxation rules (i.e., local rules $T$ that satisfy $T(\sum^m_{j=1} \alpha_j b_j) = \sum^m_{j=1} \alpha_j T(b_j)$) are Pareto optimal over all possible local taxation rules, i.e., for any (possibly nonlinear) local rule $T$, there exists a linear rule $T^{\rm lin}$ such that ${\rm PoA}(\mc{G}^n_{T^{\rm lin}}) \leq {\rm PoA}(\mc{G}^n_{T})$ and ${\rm PoS}(\mc{G}^n_{T^{\rm lin}}) \leq {\rm PoS}(\mc{G}^n_{T})$.

\begin{theorem} \label{thm:optimal_poa}
Consider the family of resource cost functions $\mc{L}=\text{span}(b_1,\dots,b_m)$ corresponding to convex, nondecreasing basis functions $b_1, \dots, b_m$, and maximum number of users $n$.  The following statements hold:

\vspace{.2cm}\noindent i)~Let $T^{\rm lin}$ denote a Pareto optimal rule in the set of all linear taxation rules.  Then, $T^{\rm lin}$ is Pareto optimal over all (possibly nonlinear) local taxation rules.

\vspace{.2cm}\noindent ii)~Let $T^{\rm PoA}$ denote the rule that minimizes the Price of Anarchy as defined in \eqref{def:minimize_poa} where $(F^{\rm PoA}_j, \rho^{\rm PoA}_j)$, $j=1,\dots,m$, are solutions to the $m$ linear programs in \eqref{linprog:optimalpoa} (one for each $b_j$).  It holds that ${\rm PoS}(\mc{G}^n_{T^{\rm PoA}}) = {\rm PoA}(\mc{G}^n_{T^{\rm PoA}})$.  Furthermore, the functions $F^{\rm PoA}_j$, $j=1,\dots,m$, are nondecreasing and unique up to rescaling.
\end{theorem}
The proof is presented in \ref{proof:optimal_poa}.  We highlight that the performance guarantee ${\rm PoS}(\mc{G}^n_{T^{\rm PoA}}) = {\rm PoA}(\mc{G}^n_{T^{\rm PoA}})$ in Theorem \ref{thm:optimal_poa}ii) is achieved by the same game instance $G\in\mc{G}^n_{T^{\rm PoA}}$.  Moreover, the instance $G$ is a simple, $n$-user game in which each user has 2 single-selection actions and there is a unique pure Nash equilibrium.  We depict the structure of the worst-case game instance as a graph in Figure~\ref{fig:singleton} where the users are represented by the edges, and the resources are the nodes.

\begin{figure}[tb]
    \centering
    \includegraphics[width=0.5\textwidth,trim={0 5.75cm 13.5cm 0}]{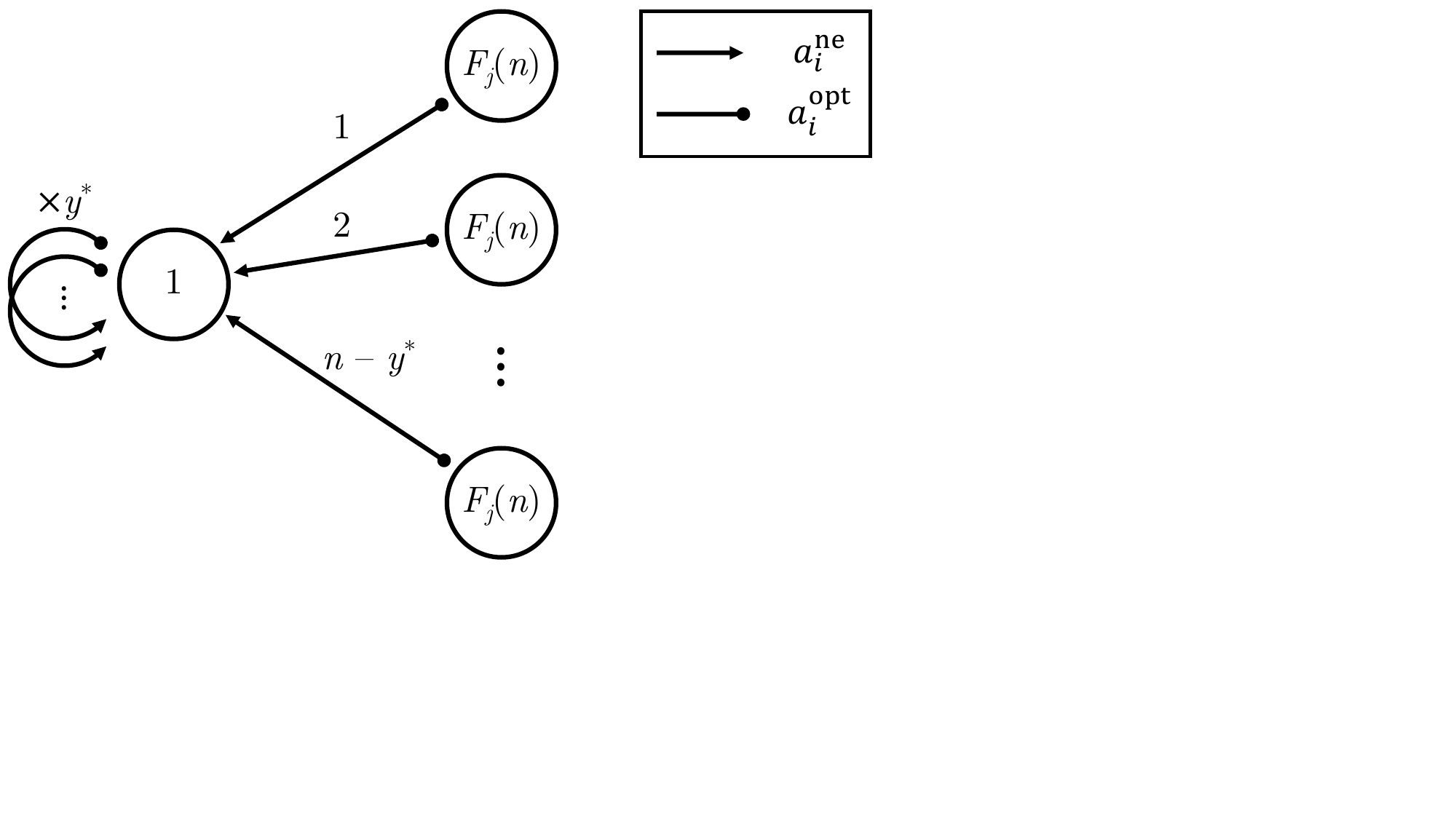}
    \caption{\emph{The worst-case game structure satisfying ${\rm PoA}(G)={\rm PoS}(G)={\rm PoA}(\mc{G}^n_{T^{\rm PoA}})={\rm PoS}(\mc{G}^n_{T^{\rm PoA}})$ in Theorem \ref{thm:optimal_poa}.}  Consider the local taxation rule $T^{\rm PoA}$ that minimizes the Price of Anarchy as described in Proposition \ref{prop:optimal_poa}.  As shown in Theorem \ref{thm:optimal_poa}, the resulting Price of Stability is equal to the minimum Price of Anarchy, i.e., ${\rm PoS}(\mc{G}^n_{T^{\rm PoA}})={\rm PoA}(\mc{G}^n_{T^{\rm PoA}})$.  We depict the worst-case game structure with the graph above where the $n$ edges are the users and the $1\leq n-y^*+1\leq n+1$ nodes are the resources.  Each resource $e\in\mc{E}$ has resource cost function $\ell_e(x) = \alpha_e \cdot b_j(x)$, where $\alpha_e\geq 0$ is the value of the node (either 1 or $F_j(n)$) and $b_j$ is one of the basis functions.  Each user $i\in N$ has two single-selection actions, i.e., $\mc{A}_i=\{\nash{a}_i,\opt{a}_i\}$.  In the above depiction, the arrow (resp. round) tip of each edge $i\in N$ indicates the resource $i$ selects in $\nash{a}_i$ (resp. $\opt{a}_i$).  Observe that all $n$ users select the left resource in the joint action $\nash{a}=(\nash{a}_1,\dots,\nash{a}_n)$ which is the unique equilibrium action since the functions $\opt{F}_j$'s are nondecreasing (by Theorem \ref{thm:optimal_poa}ii)).  In contrast, $y^*$ users select the left resource and the remaining $n-y^*$ users select individual resources in the joint action $\opt{a}=(\opt{a}_1,\dots,\opt{a}_n)$ which is the optimal action.  To obtain the result, we select $y^*\in\{0,1,\dots,n\}$ and $b_j\in\{b_1,\dots,b_m\}$ that maximize ${\rm SC}(\nash{a})/{\rm SC}(\opt{a})$.}
    \label{fig:singleton}
\end{figure}

Alternatively, one might wish to understand how unilaterally minimizing the Price of Stability impacts the Price of Anarchy.  In the previous section, we showed that the marginal cost rule is the unique local rule that achieves the minimum Price of Stability of 1 in congestion games.  \citet{paccagnan2019incentivizing} provide a tractable linear program for computing the Price of Anarchy in congestion games under the marginal cost rule.  We reproduce their result in the following proposition:

\begin{proposition}[Corollary~6.1~\cite{paccagnan2019incentivizing}] \label{prop:optimal_pos}
Consider the family of resource cost functions $\mc{L}=\text{span}(b_1,\dots,b_m)$ corresponding to convex, nondecreasing basis functions $b_1, \dots, b_m$, and maximum number of users $n$.  Then, the marginal cost rule has Price of Anarchy ${\rm PoA}(\mc{G}^n_{T^{\rm mc}})=1/\rho^{\rm mc}$, where $\rho^{\rm mc}$ is the optimal value of
\begin{equation} \label{linprog:marginal_contribution}
\begin{aligned}
    & \underset{\rho, \nu\geq 0}{maximize} \quad \rho \quad \text{subject to:} \\
    & b_j(y)y-\rho b_j(x)x + \nu[(x^2-xy)b_j(x)-x(x-1)b_j(x-1)-y(x+1)b_j(x+1)] \geq 0, \\
    & \hspace*{200pt} x,y \in \{0,\dots,n\}, \quad x+y\leq n, \quad \forall j\in\{1,\dots,m\}, \\
    & b_j(y)y-\rho b_j(x)x + \nu[xb_j(x)(2n-x-y)+(x-1)b_j(x-1)(y-n-(x+1)b_j(x+1)(x-n)] \geq 0, \\
    & \hspace*{200pt} x,y \in \{0,\dots,n\}, \quad x+y> n, \quad \forall j\in\{1,\dots,m\},
\end{aligned}    
\end{equation}
where we define $b_j(-1)=b_j(0)=b_j(n+1)=0$.
\end{proposition}
Thus, we have established that the local taxation rules that unilaterally minimize the Price of Anarchy and Price of Stability are linear and unique (Proposition \ref{prop:optimal_pos} and Theorem \ref{thm:optimal_poa}), and have provided characterizations of the corresponding Price of Stability and Price of Anarchy, respectively.

\section{The anarchy-stability trade-off in congestion games} \label{sec:bounds} 

In the previous sections, we showed that the unique local taxation rule that minimizes the Price of Anarchy has corresponding Price of Stability equal to the Price of Anarchy (Theorem \ref{thm:optimal_poa}).  Furthermore, the well-known marginal cost rule is the unique local taxation rule that minimizes the Price of Stability, always achieving a Price of Stability of 1 (Proposition \ref{prop:marginal_unique}).  Since the minimum achievable Price of Anarchy is strictly greater than 1 for all nondecreasing, convex resource costs (except constant) \cite{paccagnan2021congestion}, it immediately follows that the taxation rule that minimizes the Price of Anarchy is distinct from the taxation rule that minimizes the Price of Stability and, thus, there must exist a trade-off between these two metrics.  In this section, we develop analytical techniques for deriving upper and lower bounds on the Pareto frontier between the Price of Anarchy and Price of Stability in congestion games, which permit us to better understand the trade-off between these two metrics. Though the results in this section depend on an upper-bound, $n$, on the number of users, we discuss how these techniques can be extended to remove the dependence on $n$ in \ref{sec:arbitrary}.

\subsection{The upper bound} \label{sec:upper_bound}
Before presenting our upper bound, we introduce a modified version of the smoothness argument in \citet{christodoulou2015price} that provides upper bounds on the Price of Stability of congestion games.  Recall that all congestion games are potential games and admit the potential function $\Phi:\mc{A}\to\bb{R}$ in \eqref{eq:rosenthal_potential}.

\begin{proposition} \label{prop:smoothness_argument}
Let $\mc{G}$ denote any family of congestion games, and suppose that there exist $\zeta > 0$, $\lambda > 0$ and $\mu < 1$ such that, for every game $G\in\mc{G}$ and any two assignments $a,a' \in \mc{A}$, it holds that
\begin{equation} \label{eq:smoothness_argument}
    {\rm SC}(a) + \sum_{i\in N} [ C_i(a'_i, a_{-i})-C_i(a) ] + \zeta [ \Phi(a')-\Phi(a) ] \leq \lambda {\rm SC}(a') + \mu {\rm SC}(a).
\end{equation}
Then, the Price of Stability satisfies ${\rm PoS}(\mc{G}) \leq \lambda/(1-\mu)$.
\end{proposition}

\begin{proof}
Consider any game $G\in\mc{G}$ and let $\opt{a} \in \mc{A}$ denote an optimal assignment, i.e., ${\rm SC}(\opt{a})={\rm MinCost}(G)$.  Thus, let $\nash{a}\in{\rm NE}(G)$ denote a pure Nash equilibrium that satisfies $\Phi(\nash{a})\leq \Phi(\opt{a})$.\footnote{Observe that such a pure Nash equilibrium must always exist since any potential minimizer is a pure Nash equilibrium.}  Since $C_i(\nash{a})\leq C_i(\opt{a}_i, \nash{a}_{-i})$ for all $i\in N$ and $\Phi(\nash{a})\leq \Phi(\opt{a})$, it follows from \eqref{eq:smoothness_argument} that 
\[ {\rm SC}(\nash{a}) \leq \lambda {\rm SC}(\opt{a}) + \mu {\rm SC}(\nash{a}). \]
Rearranging the above inequality gives us that ${\rm SC}(\nash{a})/{\rm SC}(\opt{a}) \leq \lambda/(1-\mu)$.  Since $\nash{a}$ is not necessarily the pure Nash equilibrium in ${\rm NE}(G)$ with minimum social cost, it holds that ${\rm PoS}(G) \leq \lambda/(1-\mu)$, and it could hold that ${\rm PoS}(G) < \lambda/(1-\mu)$ in general.
\end{proof}

The smoothness argument in Proposition \ref{prop:smoothness_argument} provides an upper bound on the Price of Stability by bounding the efficiency of all pure Nash equilibria with potential lower than the potential at the optimal assignment.  Our next result shows how one can leverage this smoothness argument to optimize an upper bound on the Price of Stability under a maximum allowable Price of Anarchy constraint.

\begin{theorem} \label{thm:upperbound}
Consider the family of resource cost functions $\mc{L}=\text{span}(b_1,\dots,b_m)$ corresponding to basis functions $b_1,\dots,b_m$, and maximum number of users $n$.  Further, consider a maximum allowable Price of Anarchy $\bar{\Pi}\geq {\rm MinPoA}(n,\mc{L})$.  Let $\{\opt{F}_1,\dots,\opt{F}_m\},\opt{\nu}, \opt{\rho}, \opt{\gamma}, \opt{\kappa}$ be solutions to the following:\footnote{In \eqref{bilp:upperbound_finite} and forthcoming optimization problems, we reformulate the decision variables and constraints to ensure they follow a (bi)linear program structure for the reader's convenience. For example, $\nu^{-1}\rho$ is a decision variable in \eqref{bilp:upperbound_finite}. Note that the corresponding optimal values $\{\opt{F}_1,\dots,\opt{F}_m\}$, $\opt{\nu}$, $\opt{\rho}$, etc. are uniquely determined by solutions to these problems.} 
\begin{equation} \label{bilp:upperbound_finite}
\begin{aligned} 
    \underset{\{F_j\},\nu^{-1}\rho,\gamma,\nu^{-1},\kappa}{\text{maximize}} \> & \gamma \quad \text{subject to:} \\
    & \nu^{-1}\rho \geq \bar{\Pi}^{-1} \nu^{-1}, \quad \nu^{-1} \geq 0, \quad \kappa \geq 0 \\
    & \nu^{-1}b_j(y)y-\nu^{-1}\rho b_j(x)x+(x-z)F_j(x)-(y-z)F_j(x+1) \geq 0, \\
    & \hspace*{200pt} \forall (x,y,z)\in\mc{I}(n), \forall j\in\{1,\dots,m\}, \\
    & b_j(y)y - \gamma b_j(x)x + (x-z)F_j(x) - (y-z)F_j(x+1) 
        + \kappa \left[ \sum^x_{k=1} F_j(k)-\sum^y_{k=1} F_j(k) \right] \geq 0, \\
    & \hspace*{200pt} \forall (x,y,z)\in\mc{I}(n), \forall j\in\{1,\dots,m\},
\end{aligned}
\end{equation}
where we define $\mc{I}(n)$ as in \eqref{eq:In}, and $b_j(0)=F_j(0)=F_j(n+1)=0$.  Then, the local taxation rule $\opt{T}$ defined as $\opt{T}(b_j)(x)=\opt{F}_j(x)-b_j(x)$, $j=1,\dots,m$, achieves Price of Anarchy ${\rm PoA}(\mc{G}^n_{\opt{T}})=1/\opt{\rho}\leq \bar{\Pi}$ and Price of Stability ${\rm PoS}(\mc{G}^n_{\opt{T}}) \leq 1/\opt{\gamma}$. 
\end{theorem}
The proof is presented in \ref{proof:upperbound}, and amounts to reformulating the problem of computing the local taxation rule that optimizes the smoothness bound in Proposition \ref{prop:smoothness_argument} as a tractable optimization problem.  The optimization problem in \eqref{bilp:upperbound_finite} is a bilinear program with a single bilinearity, since $\kappa$ is multiplied with $F$ in the final set of constraints.  Such programs can be solved efficiently using, e.g., the method of bisections, which involves solving a finite number of linear programs for appropriate guesses of the value $\opt{\kappa}$.  

A possible interpretation of the above result is that the local rule $\opt{T}$ guarantees that every game $G\in\mc{G}^n_{\opt{T}}$ has at least one pure Nash equilibrium with social cost at most $1/\opt{\gamma}$ times greater than ${\rm MinCost}(G)$.  Recall from the proof of Proposition \ref{prop:smoothness_argument} that this equilibrium may not represent the best performing equilibrium of $G$, so this represents an upper bound on the Price of Stability, in general.

\subsection{The lower bound} \label{sec:lower_bound}
The following theorem states our corresponding lower bound on the best achievable Price of Stability for a maximum allowable Price of Anarchy $\bar{\Pi}$:

\begin{theorem} \label{thm:lower_bound}
Consider the family of resource cost functions $\mc{L}=\text{span}(b_1,\dots,b_m)$ corresponding to basis functions $b_1,\dots,b_m$, and maximum number of users $n$.  Further, consider a maximum allowable Price of Anarchy $\bar{\Pi}\geq{\rm MinPoA}(n,\mc{L})$.  Let $\opt{F}_j,\nu_j,\rho_j$ be optimal values that solve the following $m$ linear programs (one for each $j$):
\begin{equation} \label{linprog:lower_bound}
\begin{aligned}
    \underset{F,\nu^{-1},\rho\nu^{-1}}{\text{maximize}} &\> \sum^n_{x=1} F(x) \quad \text{subject to:} \\
    &\> \rho\nu^{-1} \geq \bar{\Pi}^{-1} \nu^{-1}, \quad \nu^{-1}\geq 0, \quad F(1) = 1, \\
    &\> \nu^{-1}b_j(y)y-\rho\nu^{-1}b_j(x)x+(x-z)F(x)-(y-z)F(x+1) \geq 0,
                                \forall (x,y,z)\in\mc{I}(n),
\end{aligned}
\end{equation}
where we define $\mc{I}(n)$ as in \eqref{eq:In}, and $b_j(0)=F_j(0)=F_j(n+1)=0$.  Then, the Price of Stability of any local taxation rule $T$ with ${\rm PoA}(\mc{G}^n_T)\leq \bar{\Pi}$ must satisfy ${\rm PoS}(\mc{G}^n_T)\geq \max_j\{1/\opt{\gamma}_j\}$, where
\[ \opt{\gamma}_j = \min_{0\leq v<u\leq n} \frac{b(v)v+\sum^{u-v}_{k=1} F^{(u,v)}_j(k)}{b(u)u}, \]
where $F^{(u,v)}_j(k) = \max_{v+k\leq x \leq u} \opt{F}_j(x)$ for $k=1,\dots,u-v$.
\end{theorem}
The proof is presented in \ref{proof:lower_bound}.  We highlight some important observations regarding the above result in the discussion below:

Note that the linear program in \eqref{linprog:lower_bound} must be feasible for all values $\bar{\Pi}\geq{\rm MinPoA}(n,\mc{L})$ as there exists at least one set of feasible values $F,\nu,\rho$ by Proposition \ref{prop:optimal_poa} and Theorem \ref{thm:optimal_poa}.  Furthermore, the linear program must provide a (tight) lower bound of ${\rm PoS}(\mc{G}^n_T)\geq 1$ for any $\bar \Pi$ greater than the Price of Anarchy of the marginal cost rule, ${\rm PoA}(\mc{G}^n_{T^{\rm mc}})$, since the Price of Stability of the marginal cost rule is 1.  When the basis functions are convex and nondecreasing, the linear program must also provide a (tight) lower bound ${\rm PoS}(T)\geq \bar{\Pi}$ when $\bar{\Pi}={\rm MinPoA}(n,\mc{L})$, since we showed in Part~(ii) of the proof of Theorem~\ref{thm:optimal_poa} that a worst case game in this setting has the same structure as the construction we use to obtain this lower bound.\footnote{Though the game construction we consider to obtain the result in Theorem \ref{thm:lower_bound} is of the same structure as in Figure \ref{fig:singleton}, the selection of the resources' coefficients is more nuanced in general since we have no guarantee on the monotonicity of the resource cost function $F$ in this setting.  See \ref{proof:lower_bound} for more details.}  Additionally, the game construction from which we obtain this lower bound on the Price of Stability has a unique pure Nash equilibrium $\nash{a}$ where each user $i\in N$ strictly prefers to play $\nash{a}_i$ when users $1,\dots,i-1$ play their respective actions in $\nash{a}$.  It is straightforward to verify that $\nash{a}$ is also the unique coarse-correlated equilibrium of the game and, thus, our lower bound extends to the best case coarse correlated equilibrium efficiency.

\subsection{Equilibrium efficiency in congestion games without taxes} \label{sec:efficiency_notax}

Many works have focused on identifying tight bounds on the Price of Anarchy and Price of Stability, particularly for congestion games without taxes.  In this respect, \citet{aland2011exact} put forward an expression for the Price of Anarchy in polynomial congestion games without taxes. Meanwhile, \citet{christodoulou2015price} provide exact bounds on the Price of Stability in polynomial congestion games without taxes.  As both these values are known exactly, we can compare the equilibrium performance in the absence of taxes against upper bounds on the Pareto frontier for arbitrary number of users we derive using the technique presented in \ref{sec:arbitrary}.

In Columns~2 and~3 of Table~\ref{table:notax}, we summarize the Price of Anarchy and Price of Stability bounds from the literature on polynomial congestion games without taxes.  In Columns~4 and~5, we provide upper bounds on the best achievable Price of Stability while guaranteeing Price of Anarchy no greater than that of polynomial congestion games without taxes.  Conversely, in Columns~6 and~7, we provide upper bounds on the best achievable \emph{Price of Anarchy} without exceeding the \emph{Price of Stability} in polynomial congestion games without taxes.  From the values in Table~\ref{table:notax}, one can easily verify that using no tax in the polynomial congestion games considered is not Pareto optimal.  In fact, significant improvements can be achieved in terms of Price of Anarchy while guaranteeing the same Price of Stability.  For the specific case of polynomial congestion games of degree $d=4$, we observe that the Price of Anarchy can be reduced by more than 66.7\% with no increase in the Price of Stability.  Demonstrating the suboptimality of no tax more generally remains an interesting open problem and should be considered in future work.

\begin{table}[t]
    \centering
    \caption{\emph{Comparison of Price of Anarchy and Price of Stability values for no taxes and Pareto (sub-)optimal taxation rules in congestion games with polynomial resource costs of degree $d=1,\dots,4$.} The Price of Anarchy and Price of Stability values for no taxes are summarized in Columns~2 and~3 and were derived by \citet{aland2011exact} and \citet{christodoulou2015price}, respectively.  Columns~4 and 5 (resp. 6 and 7) report values that must lie above the Pareto frontier between the Price of Anarchy and Price of Stability metrics and correspond with the Price of Anarchy (resp. Price of Stability) under no taxes.  These were computed using the technique presented in 
    \ref{sec:arbitrary}.}
    \label{table:notax}
    \begin{tabular}{|c||r|r|r|r|r|r|}
        \hline
        \multirow{2}{*}{$d$} &
            \multicolumn{2}{c|}{No taxes} &
            \multicolumn{2}{c|}{${\rm PoA}$ Requirement} &
            \multicolumn{2}{c|}{${\rm PoS}$ Requirement} \\
        & \multicolumn{1}{c|}{${\rm PoA}$} & \multicolumn{1}{c|}{${\rm PoS}$} 
        & \multicolumn{1}{c|}{${\rm PoA}$} & \multicolumn{1}{c|}{${\rm PoS}$} 
        & \multicolumn{1}{c|}{${\rm PoA}$} & \multicolumn{1}{c|}{${\rm PoS}$} \\
        \hline
        1 &     2.500 & 1.577 &   2.500 & 1.418 &  2.381 & 1.577 \\ \hline 
        2 &     9.583 & 2.361 &   9.583 & 1.156 &  7.044 & 2.361 \\ \hline
        3 &    41.536 & 3.322 &  41.536 & 1.290 & 22.930 & 3.322 \\ \hline
        4 &   267.643 & 4.398 & 267.643 & 1.135 & 88.895 & 4.398 \\ \hline
    \end{tabular}
\end{table}

\section{Attainable joint performance guarantees} \label{sec:gamebygame}

In the previous sections, we study the Price of Anarchy and Price of Stability of a given family of instances $\mc{G}_T$ as independent, worst-case measures of the equilibrium efficiency, i.e., we summarize the equilibrium efficiency of all game instances under a given taxation rule $T$ with only two numbers, ${\rm PoA}(\mc{G}_T)$ and ${\rm PoS}(\mc{G}_T)$.  Note, however, that the values ${\rm PoA}(\mc{G}_T), {\rm PoS}(\mc{G}_T)$ may not be achieved within the same game instance. Specifically, there need not exist a game instance $G\in\mc{G}_T$ such that ${\rm PoA}(G)={\rm PoA}(\mc{G}_T)$ \emph{and} ${\rm PoS}(G)={\rm PoS}(\mc{G}_T)$.  Rather, it could be that there exist two distinct games $G,G'\in\mc{G}_T$ satisfying ${\rm PoA}(G) = {\rm PoA}(\mc{G}_T)>{\rm PoA}(G')$ and ${\rm PoS}(G')={\rm PoS}(\mc{G}_T)>{\rm PoS}(G)$.  This motivates our investigation -- in this section -- of those Price of Anarchy and Price of Stability pairs that can be achieved within the same game instance, where we wish to understand if considering such \emph{attainable joint performance measure} offers more refined insights on the joint optimization of the worst and best equilibrium efficiency.

More specifically, for a given family $\mc{G}$, we aim to capture the dependence of the Price of Stability of an invidual instance on its Price of Anarchy.  To that end, for given $\tau\in[1,{\rm PoA}(\mc{G})]$, we define $\mc{G}^\tau$ as
\begin{equation} \label{eq:tau_game}
    \mc{G}^\tau := \{ G\in\mc{G}\text{ s.t. } {\rm PoA}(G)=\tau \}.
\end{equation}
Our goal is to characterize how the value ${\rm PoS}(\mc{G}^\tau)$ evolves with $\tau\in[1,{\rm PoA}(\mc{G})]$.

Our next result establishes that the tension between the Price of Anarchy and Price of Stability persists under the attainable joint performance measure.  This is based on the observation that the independently measured Price of Anarchy and Price of Stability corresponding with taxes that minimize either the Price of Anarchy or the Price of Stability are in fact attained within the same game instance. Furthermore, these instances are simple, single-selection games with unique pure Nash equilibria as described in the proof of Theorem \ref{thm:optimal_poa}ii) (in \ref{proof:optimal_poa}). We formally state these observations in the following corollary:

\begin{corollary} \label{cor:instance-by-instance}
For any family of nondecreasing, convex latency functions $\mc{L}$, and maximum number of users $n$, the following statements hold:
\begin{itemize}
    \item Let $T^{\rm PoA}$ denote a taxation rule that minimizes the Price of Anarchy of the corresponding family of instances, i.e.,
    \begin{equation}
        T^{\rm PoA} \in \underset{T}{\arg\,\min} \> {\rm PoA}(\mc{G}^n_T).
    \end{equation}
    There exists an instance $G\in\mc{G}^n_{T^{\rm PoA}}$ such that ${\rm PoA}(G)={\rm PoS}(G)={\rm PoA}(\mc{G}^n_{T^{\rm PoA}})$.
    \item Let $T^{\rm PoS}$ denote a taxation rule that minimizes the Price of Stability of the corresponding family of instances, i.e.,
    \begin{equation}
        T^{\rm PoS} \in \underset{T}{\arg\,\min} \> {\rm PoS}(\mc{G}^n_T).
    \end{equation}
    There exists an instance $G\in\mc{G}^n_{T^{\rm PoS}}$ such that ${\rm PoA}(G)={\rm PoA}(\mc{G}^n_{T^{\rm PoS}})$ and ${\rm PoS}(G)={\rm PoS}(\mc{G}^n_{T^{\rm PoS}})=1$.
\end{itemize}
\end{corollary}

Corollary \ref{cor:instance-by-instance} establishes that the extreme points of the Price of Anarchy, Price of Stability trade-off curve coincide whether we consider the independent, worst-case performance measure, or the attainable joint performance measure. It remains to be seen whether these two coincide in general, i.e., that the independent, case performance guarantee is always attainable by the same game instance for any family of instances under any taxation rule. To that end, we put forward a modified smoothness condition and game construction to characterize the relation between the Price of Anarchy and Price of Stability under the attainable joint performance measure. We show that there can be a significant separation between the independent, worst-case performance guarantee and the attainable joint performance guarantees.

\subsection{An upper bound}
We obtain an upper bound on the attainable joint performance guarantees using the following smoothness condition, which applies to any family of potential games, where each game has corresponding potential function $\Phi:\mc{A}\to\bb{R}$: 

\begin{proposition}
Given any family of congestion games $\mc{G}$ and parameter $\tau\in[1,{\rm PoA}(\mc{G})]$, suppose that there exist parameters $\kappa,\lambda_1,\lambda_2 \geq 0$, and $\mu,\nu\in\bb{R}$ such that, for every game $G\in\mc{G}^\tau$ and actions $a,a',a'' \in \mc{A}$, it holds that
\begin{equation} \label{eq:gamebygame_smoothness}
\begin{aligned}
    & \sum^n_{i=1} \Bigg[ \lambda_1[C_i(a'_i, a_{-i})-C_i(a)] 
                       +\lambda_2[C_i(a''_i, a_{-i})-C_i(a)] \Bigg]
                       +\kappa[\Phi(a'') - \Phi(a')] \\
    \leq \> & \nu \cdot {\rm SC}(a) - {\rm SC}(a') + \mu \cdot {\rm SC}(a'').
\end{aligned}
\end{equation}
Then, the Price of Stability satisfies ${\rm PoS}(\mc{G}^\tau) \leq \mu+\tau\nu$.
\end{proposition}

\begin{proof}
Consider any game $G\in\mc{G}^\tau$ and let $\opt{a} \in \mc{A}$ denote an optimal assignment, i.e., ${\rm SC}(\opt{a})={\rm MinCost}(G)$.  Let each $a^{\text{ne},1},a^{\text{ne},2}\in{\rm NE}(G)$ denote a pure Nash equilibrium of $G$, not necessarily distinct.  We let $a^{\text{ne},2}$ be a pure Nash equilibrium that satisfies $\Phi(a^{\text{ne},2})\leq \Phi(\opt{a})$.  Since $C_i(a^{\text{ne},1})\leq C_i(\opt{a}_i, a^{\text{ne},1}_{-i})$ and $C_i(a^{\text{ne},1})\leq C_i(a^{\text{ne},2}_i, a^{\text{ne},1}_{-i})$ for all $i\in N$, and $\Phi(a^{\text{ne},2})\leq \Phi(\opt{a})$, it follows from \eqref{eq:gamebygame_smoothness} that 
\[ {\rm SC}(a^{\text{ne},2}) \leq \nu \cdot {\rm SC}(a^{\text{ne},1}) + \mu \cdot {\rm SC}(\opt{a}). \]
Dividing both sides of the above inequality by ${\rm SC}(\opt{a})$, we obtain
\[ \frac{{\rm SC}(a^{\text{ne},2})}{{\rm SC}(\opt{a})} \leq \nu \cdot \frac{{\rm SC}(a^{\text{ne},1})}{{\rm SC}(\opt{a})} + \mu \leq \mu + \tau\nu, \]
where the final inequality holds since ${\rm SC}(a^{\text{ne},1})/{\rm SC}(\opt{a}) \leq {\rm PoA}(G) = \tau$ by the definition of $\mc{G}^\tau$ from \eqref{eq:tau_game}. Following the same reasoning as in the proof of Proposition \ref{prop:smoothness_argument}, it follows that ${\rm PoS}(G) \leq \mu+\tau\nu$.
\end{proof}
Observe that by using this smoothness condition, an upper bound on the attainable joint performance guarantees can be obtained for the family of instances corresponding to any class of latency functions and any taxation rule. In our next result, we use this smoothness argument to derive an upper bound on the attainable joint performance guarantees in affine and quadratic congestion games without taxes:

\begin{corollary} \label{cor:gamebygame_ub}
Consider the family of affine resource cost functions, i.e., $\mc{L}=\text{span}(b_1,b_2)$ where $b_1(x)=1$ and $b_2(x)=x$, and let $\mc{G}_0$ represent the family of affine congestion games without taxes. It holds that
\begin{equation} \label{eq:gamebygame_upperbound_affine}
    {\rm PoS}(\mc{G}^\tau_0) \leq \min\left\{ \tau, {\rm PoS}(\mc{G}_0), \frac{-4}{3}\tau+\frac{13}{3} \right\},
\end{equation}
for all $\tau\in[1,5/2]$, where ${\rm PoS}(\mc{G}_0)= = 1+\sqrt{3}/3 \approx 1.577$. Next, consider the family of quadratic resource cost functions, i.e., $\mc{L}=\text{span}(b_1,b_2,b_3)$ where $b_1(x)=1$, $b_2(x)=x$ and $b_3(x)=x^2$, and let $\mc{G}_0$ represent the family of quadratic congestion games without taxes. It holds that
\begin{equation} \label{eq:gamebygame_upperbound_quadratic}
    {\rm PoS}(\mc{G}^\tau_0) \leq \min\left\{ \tau, {\rm PoS}(\mc{G}_0), \frac{-1}{3}\tau+\frac{151}{36} \right\},
\end{equation}
for all $\tau\in[1,115/12]$, where ${\rm PoS}(\mc{G}_0)\approx 2.361$.
\end{corollary}

\begin{proof}
The proof follows from the smoothness condition in \eqref{eq:gamebygame_smoothness} by showing that the smoothness parameters $\lambda_1=\lambda_2=1$, $\kappa=3$, $\mu=13/3$ and $\nu=-4/3$ are feasible for all affine congestion games without taxes, and that the smoothness parameters $\lambda_1=1/8$, $\lambda_2=29/72$, $\kappa=35/9$, $\mu=151/36$ and $\nu=-1/3$ are feasible for all quadratic congestion games without taxes.
\end{proof}

In Figure \ref{fig:gamebygame}, we plot the upper bound on the attainable joint performance guarantees for affine congestion games without taxes provided in \eqref{eq:gamebygame_upperbound_affine} (solid black line). Observe that the upper bound demonstrates that the independent, worst-case performance guarantee, $({\rm PoA}(\mc{G}_T), {\rm PoS}(\mc{G}_T))$ -- which is $(2.500,1.577)$ for the family of affine congestion games without taxes -- cannot be achieved by any instance.  Additionally, and perhaps surprisingly, the upper bound guarantees that any affine congestion game without taxes with worst case Price of Anarchy has Price of Stability equal to 1.  Note that these two observations do not necessarily hold for every rule $T$, as we show in Corollary \ref{cor:instance-by-instance}.

\subsection{A lower bound} \label{sec:gamebygame_lowerbound}
Next, we wish to characterize a lower bound that complements the upper bound that we obtained using the smoothness condition in \eqref{eq:gamebygame_smoothness}. We provide such a lower bound by means of a game construction:  Let $a^{\rm w-ne}, a^{\rm b-ne}, a^{\rm opt}\in\mc{A}$ respectively denote the worst-case pure Nash equilibrium, best-case pure Nash equilibrium, and optimal joint allocation of a game. Observe that, for $a^{\rm w-ne}$ to be a pure Nash equilibrium, the following constraints must hold:
$$ C_i(a^{\rm w-ne}) \leq  C_i(a_i,a^{\rm w-ne}_{-i}), \forall a_i \in \mc{A}_i, \forall i \in N$$
Furthermore, to ensure that $a^{\rm b-ne}$ is a pure Nash equilibrium, it is sufficient to impose the constraints:
$$ C_i(a^{\rm b-ne}_{1:i}, a_{i+1:n}) < C_i(a^{\rm b-ne}_{1:i-i},a_{i:n}), \forall a\neq a^{\rm w-ne}\in\mc{A}, i \in N,$$
Note that $a^{\rm w-ne}$ and $a^{\rm b-ne}$ are the \emph{only} pure Nash equilibria of the game under the imposed user cost structure.

The game construction belongs to a subset of the family of instances $\mc{G}_T$ that only contains instances with at most two pure Nash equilibria ($a^{\rm w-ne}$ and $a^{\rm b-ne}$), of which $a^{\rm b-ne}$ is the game's potential minimizer.  Observe that by maximizing ${\rm SC}(a^{\rm b-ne})$ while requiring that ${\rm SC}(a^{\rm w-ne})=\tau \cdot {\rm SC}(\opt{a})$, we can obtain a lower bound on ${\rm PoS}(\mc{G}^\tau_T)$.  Since the constraints we consider impose a particular user cost structure on the games we consider, this lower bound may not necessarily be a tight characterization.  Nonetheless, the advantage of this lower bound is that -- under an appropriate parameterization -- it can be computed via linear programming methods for a given maximum number of users $n$.  We provide the details on such a parameterization and corresponding linear program in Appendix \ref{appendix:gamebygame_lowerbound} for ease of presentation.  In Figure \ref{fig:gamebygame}, we plot the lower bound on ${\rm PoS}(\mc{G}^\tau_T)$ for affine congestion games without taxes computed for a maximum of $n=4$ users (solid orange line). 

\begin{figure}[b!]
    \centering
    \caption{\emph{The attainable joint performance guarantees in affine congestion games without taxes.} We plot our upper bound (solid, black line) and lower bound (solid, orange line) on the set of feasible (PoA, PoS) pairs in the family of affine congestion games without taxes. We also plot the (PoA, PoS) of $10^5$ randomly generated instances from this family (navy ‘+’ marks), which all fall within the bounds (details on how these instances were generated are provided in the main text). Observe that the upper bound rules out any instances with joint performance equal to, or close to, the independent, worst-case performance guarantee (2.500, 1.577) (red star). Furthermore, our upper and lower bounds coincide at the point (2.50, 1.00), which implies that any worst-case affine congestion game without taxes $G$ from a PoA perspective must satisfy ${\rm PoS}(G)=1$. Finally, although examples of worst-case instances do not arise in the randomly generated instances, their distribution mimicks the shape of our bounds, i.e., high PoA corresponds with low PoS, and vice versa.}
    \includegraphics[width=0.75\textwidth,trim={0 3.5cm 5cm 0}]{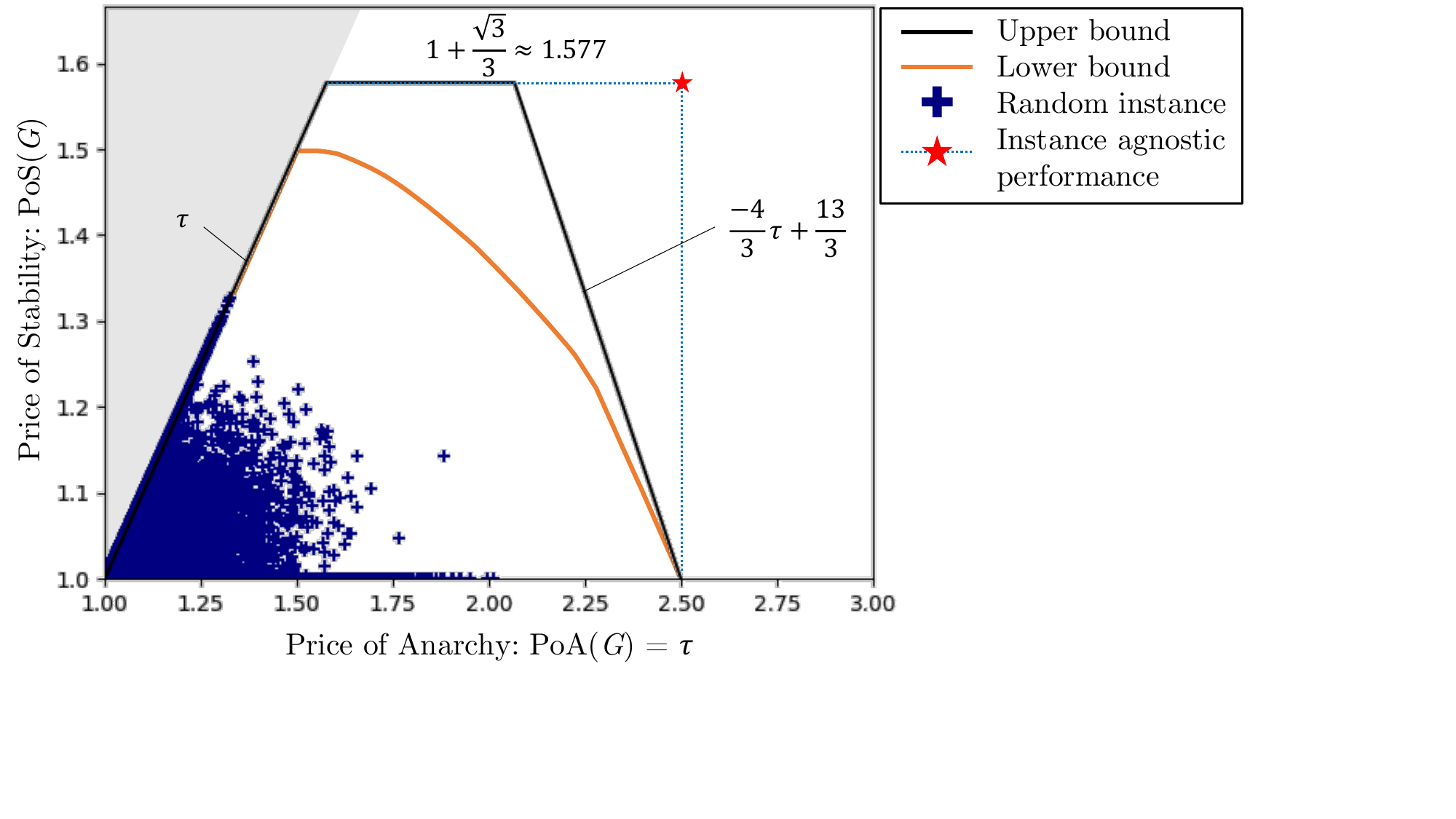}
    \label{fig:gamebygame}
\end{figure}

\subsection{Simulation results}

We provide a simulation example to compare the independent, worst-case performance guarantee and the attainable joint performance guarantees for the Price of Anarchy and Price of Stability. In our simulation example, we consider the family of affine congestion games without taxes, for which the independent, worst-case guarantee is $(2.500,1.577)$. 

Consider an affine congestion game with $n=4$ users and $|\mc{E}|=10$ resources. To each of the edges $e\in\mc{E}$, assign the resource cost function $\ell_e(x) = \alpha_e \cdot x$, where $\alpha_e$ is sampled independently from the uniform distribution between 0 and 1. Further assign to each user $i\in N$ three actions, each action consisting of the unique resources among two resources drawn (with replacement) uniformly from the set of resources.

We generate $10^5$ such random instances of affine congestion games without taxes. For each of these instances, we compute the system cost at all the pure Nash equilibria, as well as the minimum achievable system cost. From these values, we obtain the Price of Anarchy and Price of Stability of each instance. In Figure \ref{fig:gamebygame}, we plot the Price of Anarchy, Price of Stability pair for each of the $10^5$ random instances as navy blue `+' marks. Observe that the joint performance of each of the generated instances falls within our bounds on the attainable joint performance guarantees. Furthermore, though the attainable joint performance guarantees of the instances are well below the theoretical worst-case, the distribution of the instances mimicks the shape of our bounds, i.e., games with high Price of Anarchy have low Price of Stability, and vice versa.

\section{Conclusions and Future Directions}

In this paper, we investigated the consequences of minimizing the Price of Anarchy on the Price of Stability in congestion games. Our first set of results showed that the local taxation rule that minimizes the Price of Anarchy has Price of Stability equal to the Price of Anarchy. As the marginal cost rule always achieves a Price of Stability of 1, it followed that a trade-off exists between the Price of Anarchy and the Price of Stability, as the best achievable Price of Anarchy is generally strictly greater than 1. We then developed techniques for deriving upper and lower bounds on the Pareto frontier between the Price of Anarchy and Price of Stability.  Finally, we demonstrated that a similar trade-off persists when we study the attainable joint performance guarantees.  All of our results extend to the efficiency of coarse-correlated equilibria and to the Pareto frontier between the Price of Anarchy and the Price of Stochastic Anarchy.

A parallel research effort has investigated the design of global taxation rules where the design of tax functions is conditioned on all parameters of a game instance.  Interestingly, recent results have suggested that transitioning from local to global taxation rules may not provide significant reductions of the achievable Price of Anarchy \cite{paccagnan2019incentivizing}.  Whether a similar trade-off between the Price of Anarchy and the Price of Stability exists in settings with global rules is an open and interesting problem.

The Price of Anarchy and Price of Stability are only two of the many metrics used to evaluate the performance of distributed algorithms. In this respect, our contributions represent preliminary results toward the broader research agenda of identifying trade-offs in the design of taxes.  Future work should focus on understanding the impact of minimizing the Price of Anarchy on the algorithmic performance with respect to other metrics, including the rate of convergence to an equilibrium \cite{fanelli2012speed} and the transient system performance \cite{bilo2011performance,konda2021balancing}.

\citet{filos2019pareto} and \citet{ramaswamy2017impact} investigate trade-offs between the Price of Anarchy and Price of Stability in distinct classes of problems.  However, they all report findings analagous to our first main result: when the Price of Anarchy is optimized, the Price of Anarchy and Price of Stability are equal.  While it is obvious that the Price of Stability can never exceed the Price of Anarchy, it is unclear whether these two metrics must always be in tension with one another.  A relevant research direction is to understand the broader class of problems for which the Price of Anarchy can only be optimized to the detriment of the Price of Stability.

\bibliography{references}

\newpage
\begin{appendix}
    %%%%%%%%%%%%%%%%%%%%%%%%%%%%%%%%%%%%%%%%%%%%%%%%%%%%%%%%%%%%%%%%%%%%%%%%%%%%%%%%

\section{Proof of Theorem~\ref{thm:optimal_poa}} \label{proof:optimal_poa}

We prove Statements i) and ii) of the claim separately, below:

\vspace{.2cm}\noindent\emph{Proof of Statement i).} Given any rule $T$ (not necessarily linear), we show that there exists some linear rule $T^{\rm lin}$ satisfying ${\rm PoA}(\mc{G}^n_{T^{\rm lin}}) \leq {\rm PoA}(\mc{G}^n_T)$. The linear rule we consider is generated from the taxes $T(b_j)$, $j=1,\dots,m$, as follows: $T^{\rm lin}(\sum^m_{j=1} \alpha_j b_j) = \sum^m_{j=1} \alpha_j T(b_j)$. As the same set of arguments also hold for the Price of Stability, the statement follows.

Let $\bar{\mc{G}}^n_T$ and $\bar{\mc{G}}^n_{T^{\rm lin}}$ be the restricted families of congestion games with a maximum of $n$ users in which every resource $e$ has resource cost $\ell_e\in\{b_1,\dots,b_m\}$. Within this restricted class of games, the Price of Anarchy of the local rule $T$ must be equal to that of $T^{\rm lin}$ since the resulting taxes are equivalent.  For linear rules such as $T^{\rm lin}$, one can show that for any congestion game $G\in\mc{G}^n_{T^{\rm lin}}$ there is another game $G'\in\bar{\mc{G}}^n_{T^{\rm lin}}$ (possibly with many more resources) that has arbitrarily close Price of Anarchy following the proof of Theorem 5.6 in \citet{roughgarden2015intrinsic}.  In other words, the Price of Anarchy of $\mc{G}^n_{T^{\rm lin}}$ is equal to the Price of Anarchy of $\bar{\mc{G}}^n_{T^{\rm lin}}$.  Meanwhile, for general local rules such as $T$, we observe that the Price of Anarchy achieved within the restricted class of games $\bar{\mc{G}}^n_T$ can only be less than or equal to the Price of Anarchy achieved within $\mc{G}^n_T$. It immediately follows that ${\rm PoA}(\mc{G}^n_{T^{\rm lin}}) \leq {\rm PoA}(\mc{G}^n_T)$.

\vspace{.2cm}\noindent\emph{Proof of Statement ii).} Consider the following $m$ linear programs:
\begin{equation} \label{linprog:relaxedlp}
\begin{aligned}
    \underset{F,\rho}{\text{maximize}} \quad &\> \rho \\
    \text{subject to:} \quad &\> b_j(y)y-\rho b(x)x+\min\{x,n-y\}F(x)-\min\{y,n-x\}F(x+1)\geq 0 \\
                             &\> \hspace*{150pt} \forall (x,y)\in\{0,\dots,n\}\times\{1,\dots,n\} \cup (n,0).
\end{aligned}
\end{equation}
Observe that the above linear program is a relaxation of the linear program in \eqref{linprog:optimalpoa} where we only consider the constraints $(x,y,z) \in \mc{I}(n)$ such that $(x,y) \in \{0,\dots,n\}\times\{1,\dots,n\} \cup (n,0)$ and $z=\max\{0,x+y-n\}$.  Reference~\cite{paccagnan2019incentivizing} provides an expression for a set of optimal solutions $(\opt{F}_j, \opt{\rho}_j)$, $j=1,\dots,m$, to the $m$ linear programs above and show that these are also optimal solutions of the $m$ linear programs in \eqref{linprog:optimalpoa}.  As part of their proof, they show that the functions $\opt{F}_j$ must be nondecreasing.

The rest of the proof is shown in two steps: a)~we show that for the solutions $(\opt{F}_j, \opt{\rho}_j)$, $j=1,\dots,m$, to the $m$ linear programs in \eqref{linprog:optimalpoa}, the functions $\opt{F}_1, \dots, \opt{F}_m$ are unique (up to rescaling); b)~leveraging the fact that the functions $\opt{F}_1, \dots, \opt{F}_m$ are nondecreasing, we construct a congestion game $G$ that has ${\rm PoS}(G)=\max_j\{ 1/\opt{\rho}_j \}$.

\noindent\emph{Part iia) -- Proof that $\opt{F}_j$ is the unique optimal solution.} We must show that there is no other function $F$ that yields a value of $\rho=\opt{\rho}$. By contradiction, let us assume that there exists a function $\hat F$ different from $\opt{F}_j$
that also achieves $\opt{\rho}$. Let $k+1$ be the first index at which $\hat F(k+1) \neq \opt{F}_j(k+1)$. If $k=0$, due to the constraint corresponding to $(x=0,y=1)$ in the linear program in \eqref{linprog:relaxedlp}, it holds that $\hat F(1) \leq b(1) = \opt{F}_j(1)$. Since $\hat F(1) \neq \opt{F}_j(1)$, it must hold that $\hat F(1)<\opt{F}_j(1)$. A similar argument holds for $k>0$, since
\[ \hat F(k+1) \leq \max_{y \in \{1,\dots,n\}} \frac{b(y)y-\opt{\rho}b(k)k+\min\{k,n-y\}\hat F(k)}{\min\{y,n-k\}} = \opt{F}_j(k+1), \]
where the equality holds since $\hat F(k)=\opt{F}_j(k)$, by assumption.  In short, at the first $k+1$ where $\hat F$ does not equal $\opt{F}_j$, the former is always strictly lower than the latter or else a constraint in the linear program would be violated.  The contradiction follows from the constraints with $(x=n,y=y^*_n)$:
\[ \opt{\rho} \leq \frac{b(y^*_n)y^*_n+(n-y^*_n)\hat F(n)}{b(n)n} < \frac{b(y^*_n)y^*_n+(n-y^*_n)\opt{F}_j(n)}{b(n)n} = \opt{\rho}_j, \]
where the strict inequality holds since $n-y^*_n>0$.  Observe that if $n-y^*_n=0$, it holds that $\opt{\rho}_j=1$, which violates the stated conditions for uniqueness in the theorem statement.

\noindent\emph{Part iib) -- Game construction.} Here we show that for each function in $\{\opt{F}_j\}$, we can construct a congestion game that has Price of Stability equal to $1/\opt{\rho}_j$.  Without loss of generality, we assume that that all basis functions $b_1,\dots,b_m$ are scaled such that $F_j(1)=1$.  Consider the active constraint correponding to $x=n$ for each $\opt{F}_j$, which -- after some rearrangement -- appears as follows: 
\begin{equation}
    \frac{1}{\opt{\rho}_j} = \max_{y\in\{0,1,\dots,n\}} \frac{b_j(n)n}{(n-y)\opt{F}_j(n)+b_j(y)y}.
\end{equation}
Define $y^*_n \in \{0,1,\dots,n\}$ as an argument that maximizes the right-hand side in the above expression.  Consider a congestion game $G$ with a set of $n$ users $N=\{1,\dots,n\}$ and $n-y^*_n+1$ resources $\mc{E}=\{e_0,e_1,\dots,e_{n-y^*_n}\}$.  The users' action sets are defined as follows: Each user $i\in\{1,\dots,n-y^*_n\}$ has action set $\mc{A}_i=\{\nash{a}_i, \opt{a}_i\}$ where $\nash{a}_i=\{e_0\}$ and $\opt{a}_i=\{e_i\}$; and, each user $i\in\{n-y^*_n+1,\dots,n\}$ has $\mc{A}_i=\{a_i=\{e_0\}\}$.  The cost on resource $e_0$ is $b_j$, whereas each $e \in \{e_1,\dots,e_{n-y^*_n}\}$ has cost $[\opt{F}_j(n)+\epsilon] \cdot b_j$ for some $\epsilon>0$.  Since $\opt{F}_j(1)=1$ and $\opt{F}_j$ is nondecreasing, it is straightforward to verify that the assignment $(\nash{a}_1,\dots,\nash{a}_{n-y^*_n},a_{n-y^*_n+1},\dots,a_n)$ is the unique pure Nash equilibrium of the game.  Simply observe that for any assignment $a \in \Pi_i \mc{A}_i$, any user $i\in\{1,\dots,n-y^*_n\}$ selecting its action $\opt{a}_i$ can decrease its cost by selecting its action $\nash{a}_i$ instead, since $\opt{F}_j(|a|_{e_0}) < \opt{F}_j(n)+\epsilon$.  Thus, the constructed game has a unique pure Nash equilibrium, with system cost $b_j(n)n$. The assignment $(\opt{a}_1,\dots,\opt{a}_{n-y^*_n},a_{n-y^*_n+1},\dots,a_n)$ has system cost $(n-y^*_n)[\opt{F}_j(n)+\epsilon]+b_j(y^*_n)y^*_n$.  Thus, taking the limit as $\epsilon \to 0^+$, the Price of Stability of the constructed game satisfies ${\rm PoS}(G)\geq b_j(n)n / [(n-y^*_n)F(n)+b_j(y^*_n)y^*_n] = 1/\opt{\rho}_j$.  Since the function $\opt{F}_j$ has corresponding Price of Anarchy guarantee of $1/\opt{\rho}_j$, the Price of Stability must also be upper-bounded by $1/\opt{\rho}_j$.  Thus, ${\rm PoS}(\mc{G}^n_T)= \max_j \{1/\opt{\rho}_j\} = {\rm PoA}(\mc{G}^n_T)$, concluding the proof.

\section{Proof of Theorem~\ref{thm:upperbound}} \label{proof:upperbound}
Observe that the optimal upper bound achievable from the smoothness argument 
in Proposition~\ref{prop:smoothness_argument} can be computed as the solution
to the following fractional program:
\begin{align*}
    \underset{\zeta > 0, \lambda > 0, \mu}{\inf} \left\{ \frac{\lambda}{1-\mu} 
    \text{ s.t. } (\zeta,\lambda,\mu) \text{ satisfy \eqref{eq:smoothness_argument} } \forall a,a' \in \mc{A}, \forall G\in\mc{G}^n_T \right\}.
\end{align*}

To reduce the number of constraints, we introduce the following parameterization of any pair 
of assignments $a,a'\in\mc{A}$ in a game $G \in \mc{G}^n_T$:
Consider each resource $e \in \mc{E}$ and recall that 
$\ell_e(x) = \sum^m_{j=1} \alpha_{e,j} \cdot b_j(x)$ with $\alpha_{e,j}\geq 0$ for all $j$.
Let $x_e=|a|_e$, $y_e=|a'|_e$ and $z_e=|\{i\in N:e \in a_i\}\cap\{i\in N:e \in a'_i\}|$.
It follows that $(x_e,y_e,z_e)$ belongs to the set $\mc{I}(n)$ of all triplets 
$(x,y,z) \in \bb{N}^3$ that satisfy $1\leq x+y-z \leq n$ and $z \leq \min\{x,y\}$.
We define parameters $\theta(x,y,z,j) = \sum_{e\in\mc{E}_{x,y,z}} \alpha_{e,j}$ where 
$\mc{E}_{x,y,z}=\{e\in\mc{E} \text{ s.t. } (x_e,y_e,z_e) = (x,y,z) \}$.
Under this parameterization, observe that the inequality in 
\eqref{eq:smoothness_argument} can be rewritten as
\begin{equation*}
\begin{aligned}
    & \sum^m_{j=1} \sum_{x,y,z} \left[ b_j(x)x + (y-z)F_j(x+1) - (x-z)F_j(x) + \zeta\left[ \sum^x_{k=1} F_j(j) - \sum^y_{k=1} F_j(k) \right] \right] \theta(x,y,z,j) \\
    \leq \> & \sum^m_{j=1} \sum_{x,y,z} \left[ \lambda b_j(y)y + \mu b_j(x)x \right] \theta(x,y,z,j).
\end{aligned}
\end{equation*}
We note that any $(\zeta, \lambda, \mu)$ that satisfies the above constraint for 
each individual summand corresponding with the triplets $(x,y,z)\in\mc{I}(n)$ and $j=1,\dots,n$ 
must satisfy the smoothness definition in Propositon~\ref{prop:smoothness_argument}
since we have shown that the inequalities governing the smoothness definition are a linear combination 
over the $|\mc{I}(n)|\times m$ summands with nonnegative coefficients $\theta(x,y,z,j)$.
Based on this observation, we obtain the following linear program for computing an 
upper bound on the Price of Stability after the change of variables 
$\gamma = (1-\mu)/\lambda$, $\nu = 1/\lambda$ and $\kappa = \zeta/\lambda$:
\begin{align*}
    & \underset{\gamma,\nu\geq 0,\kappa\geq 0}{\text{maximize}} \quad \gamma\quad \text{subject to:} \\
    & b_j(y)y-\gamma b_j(x)x+\nu [(x-z)F_j(x)-(y-z)F_j(x+1)]
                                + \kappa \left[ \sum^x_{k=1} F_j(k)-\sum^y_{k=1}F_j(k) \right] \geq 0,\\
    & \hspace*{200pt} \forall (x,y,z)\in\mc{I}(n), \forall j\in\{1,\dots,m\}.
\end{align*}
Under this change of variables, it holds that ${\rm PoS}(\mc{G}^n_T) \leq 1/\opt{\gamma}$ for optimal solutions $(\opt{\gamma},\opt{\nu},\opt{\kappa})$.  We note that the `$\inf$' objective can now be written as a `maximize' since $\gamma \in [0,1]$ must hold.

We are interested in obtaining an upper bound on the best Price of Stability
that can be achieved by introducing taxation rules.
By including the functions $F_j$, $j=1,\dots,m$, as decision variables in the 
dual program, we obtain the following bilinear program for computing a local 
taxation rule that minimizes the upper bound on the Price of Stability:
\begin{align*}
    & \underset{\{F_j\},\gamma,\kappa\geq 0}{\text{maximize}} \quad \gamma \quad \text{subject to:} \\
    & b_j(y)y-\gamma b_j(x)x+(x-z)F_j(x)-(y-z)F_j(x+1)
                                + \kappa \left[ \sum^x_{k=1} F_j(k)-\sum^y_{k=1}F_j(k) \right] \geq 0, \\
    & \hspace*{200pt} \forall (x,y,z)\in\mc{I}(n), \forall j\in\{1,\dots,m\}.
\end{align*}
Then, for optimal solution $(\{\opt{F}_j\},\opt{\gamma},\opt{\kappa})$,
the taxation rule $\opt{T}$ defined as $\opt{T}(b_j)(x)=\opt{F}_j(x)-b_j(x)$ 
for all $j$ and $x$ satisfies ${\rm PoS}(\opt{T})\leq 1/\opt{\gamma}$.
In the above bilinear program, we have imposed $\nu=1$, which removes one 
bilinearity in the constraints.
The only remaining bilinearity involves the decision variable $\kappa$ and the 
functions $F_1, \dots, F_m$.

To obtain the local taxation rule $\opt{T}$ that guarantees a particular Price of Anarchy $\bar \Pi$ while minimizing the upper bound on the Price of Stability, we add the constraints for the Price of Anarchy from \eqref{linprog:optimalpoa} to the bilinear program. We require that $\bar \Pi$ be greater than or equal to ${\rm MinPoA}(n,\mc{L})$ for feasibility. We can then simultaneously minimize the upper bound on the Price of Stability while guaranteeing the desired Price of Anarchy.  After some rearrangement of decision variables, we obtain the bilinear program in the claim.

\section{Proof of Theorem~\ref{thm:lower_bound}} \label{proof:lower_bound}
Consider the set of games $\mc{G}^n_T$ with at most $n$ users, family of latency functions $\mc{L}=\text{span}(b_1,\dots,b_m)$ under basis functions $b_1,\dots,b_m$ and local taxation rule $T$. Define $F_j(x)=b_j(x)+T(b_j)(x)$ for $x=1,\dots,n$, $j=1,\dots,m$.  Without loss of generality, we normalize such that $F_j(1)=1$ for $j=1,\dots,m$. Define a game $G\in\mc{G}$ with $u$ users and $u-v+1$ resources for $v$ such that $0 \leq v < u \leq n$.  We denote the user set as $N=\{1,\dots,u\}$ and the resource set as $\mc{E}=\{e_0, e_1, \dots, e_{u-v}\}$. The users' action sets are defined as follows: Each user $i\in\{1,\dots,u-v\}$, has action set $\mc{A}_i = \{\nash{a}_i, \opt{a}_i\}$ with $\nash{a}_i=\{e_0\}$ and $\opt{a}_i=\{e_i\}$, while each user $i \in \{u-v+1,u\}$ has action set $\mc{A}_i = \{a_i\}$ with $a_i=\{e_0\}$. Resource $e_0$ has resource cost function $\ell_0(x)=b_j(x)$, while each resource $e_k$, $k=1,\dots,u-v$, has resource cost function $\ell_k(x)=\alpha_k b_j(x)$ where $\alpha_k=\max_{v+k \leq x \leq u} F_j(x)+\epsilon$ for $\epsilon > 0$.

Next, we prove that the game $G$ as defined above has a unique pure Nash 
equilibrium which corresponds with the assignment 
$\nash{a}=(\nash{a}_1, \dots, \nash{a}_{u-v}, a_{u-v+1}, \dots, a_u)$.
Consider the choices of user $k\in\{1,\dots,u-v\}$ with respect to any 
assignment in which all users $i \in \{1,\dots,k-1\}$ play the action 
$\nash{a}_i$.
The remaining users $i \in \{k+1,\dots,u-v\}$ play either of their actions in 
$\mc{A}_i$.
Observe that user $k$ must select either the resource $e_0$ which is currently 
selected by at least $k+l-1$ users, or the resource $e_k$ which is currently not 
selected by any other user.
It follows that user $k$ selects $\nash{a}_i=\{e_0\}$ in this scenario,
since $F_j(y) < \max_{v+k \leq x \leq u} F_j(x)+\epsilon$ with $\epsilon > 0$, 
for $y=v+k,\dots,u$.
Note that, starting from any assignment $a \in \mc{A}$,
one can repeat this argument from user $k=1$ to user $k=u-v$ to show that any 
sequence of best responses will settle on the assignment $\nash{a}$ and, thus, 
that this is the unique pure Nash equilibrium.
Note that the system cost associated with this assignment is ${\rm SC}(\nash{a})=b_j(u)u$.
Meanwhile, the system cost of the assignment $\opt{a}=(\opt{a}_1, \dots, \opt{a}_{u-v}, a_{u-v+1}, \dots, a_u)$
is ${\rm SC}(\opt{a})=b_j(v)v + \sum^{u-v}_{k=1} [ \max_{v+k\leq x \leq u} F(x) + \epsilon ]$.
Furthermore, it holds that ${\rm MinCost}(G) \leq {\rm SC}(\opt{a})$.
Thus, for $\epsilon \to 0^+$, the Price of Stability satisfies
\[ {\rm PoS}() \geq \frac{b_j(u)u}{{\rm MinCost}(G)} 
                \geq \frac{b_j(u)u}{b_j(v)v + \sum^{u-v}_{k=1} [ \max_{v+k\leq x \leq u} F_j(x) ]}. \]

We have shown that within the family of games $\mc{G}^n_T$, there exists a singleton game with a unique pure Nash equilibrium
for any $b_j$, $j=1,\dots,m$, and any pair $(u,v)$ such that $0\leq v<u\leq n$.  We also derived the lower bound on the Price of Stability for each of these games.  Observe that the maximum value of this lower bound over all $b_j$ and all valid pairs $(u,v)$ represents a lower bound on the Price of Stability, i.e.,  
\begin{align*}
    {\rm PoS}(\mc{G}^n_T) &\geq \max_j \max_{0\leq v < u \leq n} \frac{b_j(u)u}{b_j(v)v 
                        + \sum^{u-v}_{k=1} [ \max_{v+k\leq x \leq u} F_j(x) ]} \\
                    &= \max_j \max_{0\leq v < u \leq n} \frac{b_j(u)u}{b_j(v)v 
                        + \sum^{u-v}_{k=1} F^{(u,v)}_j(k)},
\end{align*} 
where we define $F^{(u,v)}_j(k) := \max_{v+k\leq x \leq u} F_j(x)$, for $k=1,\dots,u-v$, for conciseness.  It follows that, given a family of congestion games $\mc{G}^n_T$ corresponding to maximum number of users $n$, basis functions $b_1,\dots,b_m$ and local taxation rule $T$, a lower bound on the Price of Stability can be computed as ${\rm PoS}(\mc{G}^n_T) \geq \max_j \{1/\opt{\gamma}_j\}$, where $\opt{\gamma}_j$, $j=1,\dots,m$, is the optimal value of the following linear program:
\begin{align*}
    & \underset{\gamma}{\text{maximize}} \quad \gamma \quad \text{subject to:} \\
    & \gamma b_j(u)u \leq b_j(v)v + \sum^{u-v}_{k=1} F^{(u,v)}(k), \quad \forall (u,v) \in \{ (u,v)\in\bb{N}^2 \text{ s.t. } 0\leq v<u\leq n\}, \\
    & F^{(u,v)}(k) = \max_{v+k\leq x \leq u} F_j(x), \quad \forall k \in \{1,\dots,u-v\}, \forall (u,v) \in \{ (u,v)\in\bb{N}^2 \text{ s.t. } 0\leq v<u\leq n\}.
\end{align*} 
It is critical to note that we assumed $F_j(1)=1$, for $j=1,\dots,m$ in the derivation of this program.

By including the functions $F_j$, $j=1,\dots,m$, as decision variables in the above linear program, we obtain a (not necessarily convex) program for minimizing the lower bound on ${\rm PoS}(\mc{G}^n_T)$.  We can then write the following $m$ programs (one for each $b_j$) for computing the minimum lower bound on the Price of Stability achievable for a maximum allowable Price of Anarchy $\bar \Pi$ greater than or equal to the minimum achievable Price of Anarchy in $\mc{G}^n_T$, where we include the Price of Anarchy constraints from the linear program in \eqref{linprog:optimalpoa}: 
\begin{equation}
\begin{aligned}
    & \underset{F,\nu^{-1},\rho\nu^{-1},\gamma}{\text{maximize}} \quad \gamma \quad \text{subject to:} \\
    & \rho\nu^{-1} \geq \bar{\Pi}^{-1} \nu^{-1}, \quad F(1) = 1, \\
    & \nu^{-1}b_j(y)y-\rho\nu^{-1}b_j(x)x+(x-z)F(x)-(y-z)F(x+1) \geq 0, \quad \forall (x,y,z)\in\mc{I}(n), \\
    & \gamma b_j(u)u \leq b_j(v)v + \sum^{u-v}_{k=1} F^{(u,v)}(k), \quad \forall (u,v) \in \{ (u,v)\in\bb{N}^2 \text{ s.t. } 0\leq v<u\leq n\}, \\
    & F^{(u,v)}(k) = \max_{v+k\leq x \leq u} F(x), \quad \forall k \in \{1,\dots,u-v\}, \forall (u,v) \in \{ (u,v)\in\bb{N}^2 \text{ s.t. } 0\leq v<u\leq n\}. 
\end{aligned}
\end{equation}
Let $\opt{F}_j,\opt{\nu}_j,\opt{\rho}_j,\opt{\gamma}_j$ be the optimal values that solve the above program.  The Price of Anarchy achieved by the corresponding local taxation rule is ${\rm PoA}(\mc{G}^n_{\opt{T}})=\max_j\{1/\opt{\rho}_j\}$ and the resulting optimal lower bound on the Price of Stability is ${\rm PoS}(\mc{G}^n_{\opt{T}}) \geq \max_j\{1/\opt{\gamma}_j\}$. 

Note that the above program is not a convex program because of the equality constraints for the values $F^{(u,v)}(k)$.  Next, we show that solving the above program is equivalent to the problem of maximizing $\sum^n_{x=1} F(x)$ subject to the Price of Anarchy constraints.  First, for each basis function $b_j$, observe that maximizing the value of $\gamma$ is equivalent to maximizing the sum over values $F^{(u,v)}(k)$, for $k=1,\dots,u-v$, for all $(u,v)$ such that $0\leq v<u\leq n$. Second, observe that the upper bound on $F(x+1)$ imposed by $F(x)$ and $\rho$ in the constraints corresponding to the Price of Anarchy is increasing in the value of $F(x)$, i.e., 
\[ (y-z) F(x+1) \leq (x-z)F(x) - \rho\nu^{-1}b_j(x)x + \nu^{-1}b_j(y)y, 
    \quad \forall (y,z) \in \{ (y,z) \in \bb{N}^2 \text{ s.t. } (x,y,z) \in \mc{I}(n) \}, \]
since it always holds that $x-z\geq 0$ since $z\leq\min\{x,y\}$ in the definition of $\mc{I}(n)$.  Thus, maximizing the value of $F(x+1)$ corresponds with maximizing the value of $F(x)$.  Finally, maximizing the sum over values $F^{(u,v)}(k)$, for $k=1,\dots,u-v$, for all $(u,v)$ such that $0\leq v<u\leq n$, is equivalent to maximizing the sum over values $\max_{v+k\leq x \leq u} F(x)$, by definition.  We showed above that maximizing the feasible value of any given decision variable $F(\hat x)$, $1 \leq \hat x \leq n$, corresponds with maximizing the values of $F(x)$, for $x=1,\dots,\hat x-1$.  Thus, maximizing $\max_{v+k\leq x \leq u} F(x)$, for $k=1,\dots,u-v$, for all $(u,v)$ such that $0\leq v<u\leq n$, is equivalent to maximizing $F(k)$, $k=1,\dots,n$, which is equivalent to maximizing $\sum^n_{k=1} F(k)$.  It follows that the nonconvex program above yields the same solution as the linear program in \eqref{linprog:lower_bound}, concluding the proof.

\section{Arbitrary number of users} \label{sec:arbitrary}

In this section, we seek to investigate the system performance that can be obtained using a local taxation rule in settings with an arbitrary number of users.  We first observe that a lower bound on the Pareto curve follows immediately from the linear program in Section~\ref{sec:bounds}, as the Price of Anarchy and Price of Stability metrics are both nondecreasing in the maximum number of users $n$.  As the same argument does not apply to upper bounds on the two metrics, we develop a method for computing an upper bound on the Pareto curve between the Price of Anarchy and Price of Stability in polynomial congestion games with any number of users.

Before presenting our next result, we define several parameters that are necessary for extending the solution of the upcoming finite dimensional bilinear program to an arbitrary number of users.  Given integer $d\in\bb{N}_{\geq 1}$, even integer $\bar{n}\in\bb{N}_{\geq 2}$, optimal values $\opt{\nu},\opt{\rho},\opt{\gamma},\opt{\kappa}$ and function $\opt{F}:\{1,\dots,n\}\to\bb{R}$, we define:
\begin{equation}
\label{eq:definitionsFRGMarg}
\begin{aligned}
    F^\infty(x)                   &:= \begin{cases}
        \opt{F}(x) \quad & \text{if } 0 \leq x \leq \bar{n}/2, \\
        \beta [x^{d+1}-(x-1)^{d+1}] \quad & \text{if } x > \bar{n}/2
    \end{cases} \\
    \rho^\infty     &:= \min\Big\{ \opt{\rho}, \beta\opt{\nu}\frac{\bar{n}}{2}\Big[1-\Big(1-\frac{2}{\bar{n}}\Big)^{d+1}\Big] 
                                                - d \Big[\frac{\beta\opt{\nu}}{d+1} \frac{\bar{n}}{2}\Big(\Big(\frac{2}{\bar{n}}+1\Big)^{d+1}-1\Big)\Big]^{1+1/d} \Big\} \\
    \gamma^\infty   &:= \min\{ \opt{\gamma}, \gamma_1, \gamma_2 + \beta\opt{\kappa}, \gamma_3 \} 
\end{aligned}
\end{equation}
where $\beta:= \opt{F}(\bar{n}/2)/[(\bar{n}/2)^{d+1}-(\bar{n}/2-1)^{d+1}]$,
\begin{equation}\label{Margeq:gamma1}
\begin{aligned}
    \gamma_1 &:= \min_{x \in \{1,\dots,\bar{n}/2-1\}}  \\
    & \frac{1}{x^{d+1}} 
        \Big[ \hat{y}_1^{d+1}(x)
            + \opt{F}(x)x - \opt{F}(x+1)\hat{y}_1(x) - \opt{\kappa} \!\!\! \sum^{\bar{n}/2}_{k=x+1} \opt{F}(k) 
            - \beta\opt{\kappa} \Big[  \hat{y}^{d+1}(x) - \Big(\frac{\bar{n}}{2}\Big)^{d+1} \Big] \Big]  
\end{aligned}
\end{equation}
where $\hat{y}_1(x) = \max\{ \bar{n}/2+1, [\opt{F}(x+1)/(d+1)/(1-\beta\opt{\kappa})]^{1/d}\}$,
\begin{equation} \label{Margeq:gamma2}
\begin{aligned}
    \gamma_2 &:= \min_{y\in\{0,\dots,\bar{n}/2-1\}} \min_{r\geq \bar{n}/2} \\
        & \beta y + \frac{1}{r^{d+1}} \Big[ y^{d+1}
        + \beta[r^{d+1}-(r{-}1)^{d+1}]r - \beta (r{+}1)^{d+1} y
        + \opt{\kappa}\!\!\!  \sum^{\bar{n}/2-1}_{k=y+1}\opt{F}(k) - \beta\opt{\kappa}\Big(\frac{\bar{n}}{2}{-}1\Big)^{d+1} \Big] 
\end{aligned}
\end{equation}
\noindent and,
\begin{equation} \label{Margeq:gamma3}
    \gamma_3 := \min_{x \geq \bar{n}/2}
    \beta\opt{\kappa} + \frac{1}{x^{d+1}} \Big[ (1-\beta\opt{\kappa})\hat{y}_3^{d+1}(x) + \beta[x^{d+1}-(x-1)^{d+1}]x - \beta[(x+1)^{d+1}-x^{d+1}]\hat{y}_3(x) \Big] 
\end{equation}
where $\hat{y}_3(x) = \max\{ \bar{n}/2, [\beta((x+1)^{d+1}-x^{d+1})/(d+1)/(1-\beta\opt{\kappa})]^{1/d} \}$.

The following theorem presents our upper bound on the Pareto frontier between the Price of Anarchy and Price of Stability metrics in polynomial congestion games for any number of users:

\begin{theorem} \label{thm:Marg_arbitrary_users}
Consider the family of resource cost functions $\mc{L}=\text{span}(b_1)$ corresponding to basis function $b_1(x)=x^d$, of order $d\in\bb{N}_{\geq 1}$.  Further, consider a maximum allowable Price of Anarchy $\bar{\Pi}^{-1} \geq \sup_{n\geq1} {\rm MinPoA}(n,\mc{L})$.  Let $\opt{F},\opt{\nu},\opt{\rho},\opt{\gamma},\opt{\kappa}$ be optimal values that solve the following bilinear program for some even integer $\bar{n}\in\bb{N}_{\geq 2}$ and $\epsilon > 0$:
\begin{align}
    & \hspace{-25pt} \underset{F\geq 0,\nu^{-1}\geq 0,\nu^{-1}\rho,\gamma,\kappa\geq 0}{\text{maximize}} \quad \gamma \quad \text{subject to:} \nonumber \\
    & \nu^{-1}\rho \geq ({\rm PoA}^*)^{-1} \nu^{-1} \label{Margconstr:poa_feasible} \\
    & F(x+1) \geq F(x), 
        && \forall x\in\{1,\dots,\bar n-1\} \label{Margconstr:nondecreasing} \\
    & F\Big(\frac{\bar{n}}{2}\Big) \leq \nu^{-1}\Big[\Big(\frac{\bar{n}}{2}+1\Big)^{d+1}-\Big(\frac{\bar{n}}{2}\Big)^{d+1}\Big] \label{Margconstr:below_mc} \\
    & F(1)+ \frac{(d+1)(\bar{n}/2+1)^d\kappa}{(\bar{n}/2)^{d+1}-(\bar{n}/2-1)^{d+1}} F\Big(\frac{\bar{n}}{2}\Big) 
        \leq (d+1)\Big(\frac{\bar{n}}{2}+1\Big)^d \label{Margconstr:yStar} \\
    & 1 -  \frac{\kappa}{(\bar{n}/2)^{d+1}-(\bar{n}/2-1)^{d+1}} F\Big(\frac{\bar{n}}{2}\Big) \geq \epsilon \label{Margconstr:bklessthan1}\\
    & \nu^{-1} y^{d+1} - \nu^{-1}\rho x^{d+1} + xF(x) - yF(x+1) \geq 0, 
        && \forall (x,y)\in\mc{I}_{\leq \bar{n}}(\bar{n}) \label{Margconstr:poa_lessthann} \\
    & y^{d+1} - \gamma x^{d+1} + xF(x) - yF(x+1) + \kappa \Big[ \sum^x_{k=1} F(k) - \sum^y_{k=1} F(k) \Big]\geq 0, 
        && \forall (x,y)\in\mc{I}_{\leq \bar{n}}(\bar{n}) \label{Margconstr:upos_lessthann} 
\end{align}
where $\mc{I}_{\leq n}(n)$ is the set of all pairs $x,y\in\{0,\dots,n\}$ such that $1 \leq x+y \leq n$.  Then, for $F^\infty:\bb{N}\to\bb{R}$, $\rho^\infty$ and $\gamma^\infty$ defined as in \eqref{eq:definitionsFRGMarg}, the local taxation rule $T^\infty$ with $T^\infty(x^d)(x) = F^\infty(x)-x^d$ satisfies ${\rm PoA}(\mc{G}_{T^\infty})\leq 1/\rho^\infty$ and ${\rm PoS}(\mc{G}_{T^\infty})\leq 1/\gamma^\infty$. 
\end{theorem}

Though the formal proof of Theorem~\ref{thm:Marg_arbitrary_users} is presented in \ref{proof:Marg_arbitrary_users}, here we provide an informal outline for the reader's convenience.  A similar approach can be used to compute upper bounds for any class of congestion games.  For congestion games with resource costs in $\mc{L}=\text{span}(x^d)$, $d\in\bb{N}_{\geq 1}$, and any number of users $n$, the proof amounts to showing that the values $\gamma^\infty, \rho^\infty, \opt{\nu}, \opt{\kappa}$ from the bilinear program in Theorem~\ref{thm:Marg_arbitrary_users} satisfy the constraints of two linear programs governing upper bounds on the Price of Anarchy and the Price of Stability, respectively. The constraints of these two linear programs are parameterized by each pair $x,y\in\{0,\dots,n\}$. We divide the proof as follows:

\vspace{.2cm}\noindent \emph{-- Upper bound on the Price of Anarchy:} 
In the first part, we show that the values $\rho^\infty,\opt{\nu}$ are feasible points of the linear program in \eqref{linprog:optimalpoa} for the function $F^\infty$ and $n$ users.  We first focus on the values $x,y$ such that $1\leq x+y\leq n$.  We show that the constraints parameterized by $0\leq x<\bar{n}/2$ and $y\geq 0$ are equivalent to constraints from the bilinear program in Theorem~\ref{thm:Marg_arbitrary_users}, leveraging the fact that $F^\infty(x)=\opt{F}(x)$ and $F^\infty(x+1)=\opt{F}(x+1)$.  Then, for $x\geq \bar{n}/2$ and $y\geq 0$, we prove that all the constraints are satisfied if $\rho^\infty$ is less than or equal to the second expression in the minimum that governs the definition of $\rho^\infty$ in \eqref{eq:definitionsFRGMarg}.  Finally, we show that the constraints with $x+y>n$ are redundant, as they are less strict than those with $1\leq x+y\leq n$.

\vspace{.2cm}\noindent \emph{-- Upper bound on the Price of Stability:}
Next, we show that the values $\gamma^\infty,\opt{\kappa}$ are feasible points of the 
linear program in Section~\ref{sec:upper_bound} for $\nu=1$, the function $F^\infty$ 
and $n$ users.
We first focus on the values $x,y$ such that $1\leq x+y\leq n$.
We show that the constraints with $0\leq x<\bar{n}/2$ and $0\leq y\leq\bar{n}/2$, 
and the constraints with $x=0$ and $y>\bar{n}/2$, are equivalent 
to constraints from the bilinear program in Theorem~\ref{thm:Marg_arbitrary_users}.
Then, we prove that $\gamma^\infty,\opt{\kappa}$ satisfy the constraints 
with $1\leq x< \bar{n}/2$ and $y>\bar{n}/2$ because $\gamma^\infty \leq \gamma_1$,
the constraints with $x\geq \bar{n}/2$ and $0\leq y<\bar{n}/2$ because $\gamma^\infty\leq \gamma_2$,
and the constraints with $x,y\geq \bar{n}/2$ because $\gamma^\infty\leq \gamma_3$, 
for $\gamma_1,\gamma_2,\gamma_3$ defined as in \eqref{Margeq:gamma1}--\eqref{Margeq:gamma3}.
Finally, we show that the constraints with $x+y>n$ are less strict than those with $1\leq x+y\leq n$.

\vspace{.2cm}

In the study of polynomial congestion games, we are often interested in the setting 
where resource cost functions have the form $\ell(x) = \sum^d_{j=0} \alpha_j x^j$ for $\alpha_j\geq 0$
for given order $d\geq 1$.
As it is stated, the result in Theorem~\ref{thm:Marg_arbitrary_users} can only accommodate resource cost functions corresponding to a single monomial basis function.
However, consider the scenario where we solve the bilinear program program in Theorem~\ref{thm:Marg_arbitrary_users} for each monomial basis functions $b_1,\dots,b_m$, under the same values $\kappa=\kappa^*$, $\nu=\nu^*$ and $\bar{\Pi}^{-1}$ greater than or equal to $\sup_{n\geq 1} {\rm MinPoA}(n,\mc{L}=\text{span}(b_1,\dots,b_m))$.  Then, the values ${\rm PoA}(\mc{G}_T) \leq \max_j \{1/\rho^\infty_j\}$ and ${\rm PoS}(\mc{G}_T) \leq \max_j \{1/\gamma^\infty_j\}$ must be valid upper bounds on the Pareto frontier between the Price of Anarchy and Price of Stability in $\mc{G}$, where $\rho^\infty_j$ and $\gamma^\infty_j$ are derived as in \eqref{eq:definitionsFRGMarg} for all $j=1,\dots,m$.  We state this consequence of Theorem~\ref{thm:Marg_arbitrary_users} in the following corollary:

\begin{corollary}
Consider the family of resource cost functions $\mc{L}=\text{span}(b_1,\dots,b_m)$ corresponding to basis function $b_j(x)=x^{d_j}$ of order $d_j\in\bb{N}_{\geq 1}$.  Further, consider a maximum allowable Price of Anarchy $\bar{\Pi} \geq \sup_{n\geq 1} {\rm MinPoA}(n,\mc{L})$.  Given even integer $\bar{n}\in\bb{N}_{\geq 2}$, $\hat \kappa \geq 0$, $\hat \nu\geq 0$ and $\epsilon > 0$, for each basis functions $b_j$, $j=1,\dots,m$, let $\opt{F}_j,\opt{\rho}_j,\opt{\gamma}_j$ denote optimal values corresponding to a solution of the corresponding bilinear program in Theorem~\ref{thm:Marg_arbitrary_users} under additional equality constraints $\kappa_j=\hat \kappa$ and $\nu_j=\hat \nu$.  Then, for $F^\infty_j:\bb{N}\to\bb{R}$, $\rho^\infty_j$ and $\gamma^\infty_j$ defined as in \eqref{eq:definitionsFRGMarg}, the local taxation rule $T^\infty$ with $T^\infty(b_j)(x) = F^\infty_j(x)-x^{d_j}$ satisfies ${\rm PoA}(\mc{G}_{T^\infty})\leq \max_j\{1/\rho^\infty_j\}$ and ${\rm PoS}(\mc{G}_{T^\infty})\leq \max_j\{1/\gamma^\infty_j\}$.
\end{corollary}

\section{Proof of Theorem~\ref{thm:Marg_arbitrary_users}} \label{proof:Marg_arbitrary_users}
It is important to note that $F^\infty$ is a nondecreasing function by the constraints in \eqref{Margconstr:nondecreasing} and by its definition in \eqref{eq:definitionsFRGMarg}. Furthermore, there is always a feasible point in the bilinear program since $\bar \Pi$ is greater than or equal to the minimum achievable Price of Anarchy and because monomials of order $d\geq 1$ are all convex and nondecreasing.  One feasible point corresponds with $\kappa=0$, $\nu=1$, and $F$ and $\rho$ solving the linear program in Proposition~\ref{prop:optimal_poa} for $\bar{n}$ users.  Observe that all the constraints in the bilinear program are satisfied because 
the function $F$ is unique and nondecreasing by Theorem~\ref{thm:optimal_poa}, $\rho \geq \bar{\Pi}^{-1}$ since $1/\rho$ is the minimum achievable Price of Anarchy (which is strictly greater than 1 for polynomials~\cite{bilo2019dynamic}) and because the tax $T(x^d)$ must be lower (pointwise) than the marginal contribution.  The last statement can be shown by virtue of our result on the best achievable lower bound on the Price of Stability in Theorem~\ref{thm:lower_bound} which showed that the lower bound decreases for larger $F$.  Since the Price of Stability of $T$ according to the lower bound will be $1/\rho > 1$ for $\bar{n}\geq 2$ 
and the Price of Stability of marginal contribution is 1, our statement must hold.

The following inequalities are useful for the proof:

Observe that, for any $x\geq \bar{n}/2$, it holds that
\begin{equation} \label{ident:xplusonetothed}
\begin{aligned}
    (x+1)^d = \sum^d_{k=0} \binom{d}{k} x^{d-k}
            \leq x^d+x^{d-1} \sum^d_{k=1} \binom{d}{k} \Big( \frac{\bar{n}}{2} \Big)^{k-1}
            = x^d+x^{d-1} \frac{\bar{n}}{2} \Big[ \Big(\frac{2}{\bar{n}}+1\Big)^d-1 \Big].
\end{aligned}
\end{equation}

We will also use the following two inequalities for sums of polynomials 
with $d\geq 1$ and $x\geq y > 0$:
\begin{align}
    \sum^x_{k=y+1} k^d &\geq \frac{1}{d+1} (x^{d+1}-y^{d+1}) + \frac{1}{2} (x^d-y^d) \label{ident:sum_greaterthan} \\
    \sum^x_{k=y+1} k^d &\leq (x-y)x^d \leq x^{d+1} - y^{d+1}. \label{ident:sum_lessthan}
\end{align}
The remainder of the proof is divided into two parts as in the informal outline 
in Section~\ref{sec:arbitrary}:

\noindent\emph{-- Upper bound on the Price of Anarchy:} 
Consider the following linear program for computing the Price of Anarchy for any arbitrary
number of users $n$ given a nondecreasing function $F$:
\begin{equation} \label{eq:linprog_poafeasible}
\begin{aligned}
    & \underset{\rho,\nu\geq 0}{\text{maximize}} \quad \rho \quad \text{subject to:} \\
    & y^{d+1}-\rho x^{d+1} + \nu [xF(x)-yF(x+1)] \geq 0, \quad \forall x,y \in \{0,\dots,n\} \text{ s.t. } 1 \leq x+y \leq n, \\
    & y^{d+1}-\rho x^{d+1} + \nu [(n-y)F(x)-(n-x)F(x+1)] \geq 0, \quad \forall x,y \in \{0,\dots,n\} \text{ s.t. } x+y > n.
\end{aligned}
\end{equation}
First, we show that the above linear program is identical to the linear program 
in \eqref{linprog:optimalpoa} for $F$ nondecreasing.
Observe that the constraints from the linear program in Section~\ref{sec:upper_bound} are
\begin{align*}
    & y^{d+1}-\gamma x^{d+1}+(x-z)F(x)-(y-z)F(x+1) \\
    =\> & y^{d+1}-\gamma x^{d+1}+xF(x)-yF(x+1) + z[F(x+1)-F(x)].
\end{align*}
Since $F(x+1)-F(x)\geq 0$, the above expression is minimized for the smallest value of $z$.
If follows that $z=0$ when $x+y\leq n$ and $z=x+y-n$ when $x+y> n$, since the triplets 
$(x,y,z) \in \mc{I}(n)$ satisfy $1\leq x+y-z\leq n$ and $z\leq \min\{x,y\}$.
Thus, for $x,y\in\{0,\dots,n\}$, $x-z=x$ and $y-z=y$ when $x+y\leq n$, while $x-z=n-y$ and $y-z=n-x$ when $x+y>n$.

We show that the values $(\rho,\nu)=(\rho^\infty,\opt{\nu})$ 
are feasible in the above linear program for $F=F^\infty$ as defined in the claim
and arbitrary $n$.
We dispense with the case where $x=0$ as, in this case,
the strictest constraint on $F^\infty(1)=\opt{F}(1)$ is at $(x,y)=(0,1)$,
which is already included by the constraints in \eqref{Margconstr:poa_lessthann}.
We first consider the constraints with $1\leq x+y\leq n$.

\noindent-- In the region where $1\leq x<\bar{n}/2$ and $0\leq y\leq \bar{n}/2$,
observe that $x+y<\bar{n}$, $F^\infty(x)=\opt{F}(x)$ and $F^\infty(x+1)=\opt{F}(x+1)$.
Then, the values $(\rho^\infty,\opt{\nu})$ are feasible in the linear program in \eqref{eq:linprog_poafeasible} 
for $F=F^\infty$ because $(\opt{F},\opt{\rho},\opt{\nu})$ satisfy the constraints in \eqref{Margconstr:poa_lessthann}, 
$F^\infty(x)=\opt{F}(x)$ for $x\leq \bar{n}/2$ and $\rho^\infty\leq\opt{\rho}$.

\noindent-- Observe that, in the region where $1\leq x<\bar{n}/2$ and $y>\bar{n}/2$,
the constraint are less strict than when $y=\bar{n}/2$ if it holds that
\begin{align*}
    & \Big(\frac{\bar{n}}{2}\Big)^{d+1}-\frac{\bar{n}}{2}\opt{\nu}\opt{F}(x+1) \leq y^{d+1}-y\opt{\nu}\opt{F}(x+1) \\
    \iff & \frac{y^{d+1}-(\bar{n}/2)^{d+1}}{y-\bar{n}/2} \geq \opt{\nu}\opt{F}(x+1).
\end{align*}
The left-hand side of the last line is minimized for $y=\bar{n}/2+1$ by convexity 
and is most constraining for $x=\bar{n}/2-1$ since $\opt{F}$ is nondecreasing.
Observe that this condition on $\opt{F}(\bar{n}/2)$ holds by the constraint in \eqref{Margconstr:below_mc}.

\noindent-- Consider the region where $x\geq \bar{n}/2$ and $y\geq 0$.
In this scenario, the constraints read as
\[ y^{d+1}-\rho^\infty x^{d+1}+\beta\opt{\nu} [x^{d+1}-(x-1)^{d+1}]x-\beta\opt{\nu} [(x+1)^{d+1}-x^{d+1}]y \geq 0. \]
Observe that the left-hand side in the above is convex in $y$ and that it is minimized over 
the nonnegative reals $y\geq 0$ at $\hat{y}=[\beta\opt{\nu} [(x+1)^{d+1}-x^{d+1}]/(d+1)]^{1/d}$.
Thus, it is sufficient to show that the following holds:
\begin{align*}
    & \rho^\infty \leq \beta\opt{\nu} \Big[x-\frac{(x-1)^{d+1}}{x^d}\Big]
        -d\frac{1}{(d+1)^{1+\frac{1}{d}}}[\beta\opt{\nu} [(x+1)^{d+1}-x^{d+1}]]^{1+\frac{1}{d}}\frac{1}{x^{d+1}} \\
    \impliedby & \rho^\infty\leq \beta\opt{\nu} x\Big[1-\Big(1-\frac{1}{x}\Big)^{d+1}\Big]
        -d\Big[\frac{\beta\opt{\nu}}{d+1}\frac{\bar{n}}{2}\Big(\Big(\frac{2}{\bar{n}}+1\Big)^{d+1}-1\Big)\Big]^{1+\frac{1}{d}},
\end{align*}
where the implication holds by the identity in \eqref{ident:xplusonetothed}.
The above inequality is strictest for $x=\bar{n}/2$ and is satisfied by the definition of $\rho^\infty$ in \eqref{eq:definitionsFRGMarg}.

Note that, in the above, we have shown that for any $x,y \geq 0$, the linear program constraints 
corresponding with $1\leq x+y\leq n$ are satisfied, without ever explicitly using the fact that 
$1\leq x+y\leq n$.
Here, we show that $(\rho^\infty,\opt{\nu})$ is feasible for the constraints with 
$x+y > n$ by observing that these are less strict than the constraints we 
have already shown to be satisfied.
Observe that this amounts to showing that
\[ \nu xF^\infty(x) - \nu yF^\infty(x+1) \leq \nu (n-y)F^\infty(x) - \nu (n-x)F^\infty(x+1) \iff \nu (x+y-n)[F(x)-F(x+1)] \leq 0. \]
This must hold, since $\nu\geq 0$, $x+y-n>0$ and $F^\infty$ is nondecreasing.

\noindent\emph{-- Upper bound on the Price of Stability:}
We continue by proving that $(\gamma,\nu,\kappa)=(\gamma^\infty,1,\opt{\kappa})$ are feasible 
in the following linear program for $F=F^\infty$ as defined in the claim and arbitrary $n$:
\begin{equation} \label{eq:linprog_uposfeasible}
\begin{aligned}
    & \underset{\gamma,\nu\geq 0,\kappa\geq 0}{\text{maximize}} \quad \gamma \quad \text{subject to:} \\
    & y^{d+1}-\gamma x^{d+1} + \nu [xF(x)-yF(x+1)] 
        + \kappa \left[ \sum^x_{k=1} F(k) - \sum^y_{k=1} F(k) \right] \geq 0, \> \forall (x,y)\in\mc{I}_{\leq n}(n), \\
    & y^{d+1}-\gamma x^{d+1} + \nu [(n-y)F(x)-(n-x)F(x+1)] 
        + \kappa \left[ \sum^x_{k=1} F(k) - \sum^y_{k=1} F(k) \right] \geq 0, \> \forall (x,y)\in\mc{I}_{>n}(n).
\end{aligned}
\end{equation}
where $\mc{I}_{\leq n}(n)$ and $\mc{I}_{>n}(n)$ are the sets of pairs $x,y\in\{0,\dots,n\}$ such that $x+y\leq n$ and $x+y>n$, respectively.
Following an identical set of arguments as the ones we used to show the equivalence 
of the linear programs in \eqref{linprog:optimalpoa} and \eqref{eq:linprog_poafeasible}
for $F$ nondecreasing, one can verify that the linear program in Section~\ref{sec:upper_bound} is 
identical to the above linear program for $F$ nondecreasing.

We first show that that $(\gamma,\nu,\kappa)=(\gamma^\infty,1,\opt{\kappa})$ is feasible  
for constraints with $1\leq x+y\leq n$.

\noindent-- In the region where $0\leq x< \bar{n}/2$ and $0\leq y\leq \bar{n}/2$, 
observe that $x+y<\bar{n}$, $F^\infty(x)=\opt{F}(x)$ and $F^\infty(x+1)=\opt{F}(x+1)$.
Then, the values $(\gamma^\infty,1,\opt{\kappa})$ are feasible in the linear program in 
\eqref{eq:linprog_uposfeasible} because $(\opt{F},\opt{\gamma},1,\opt{\kappa})$
satisfy the constraints in \eqref{Margconstr:upos_lessthann} and $\gamma^\infty\leq\opt{\gamma}$.

\noindent --In the setting where $x=0$ and $y\geq\bar{n}/2$, the following must hold:
\[ (1-\beta\opt{\kappa})y^{d+1}-\opt{F}(1)y-\opt{\kappa}\sum^{\bar{n}/2}_{k=1}\opt{F}(k)+\beta\opt{\kappa}\Big(\frac{\bar{n}}{2}\Big)^{d+1}\geq 0. \]
Since $1-\beta\opt{\kappa}>0$, it follows that the left-hand side is convex and is minimized 
over nonnegative real values $y\geq\bar{n}/2$ by
\[ \hat{y}=\max\Big\{\frac{\bar{n}}{2},\Big[\frac{\opt{F}(1)}{(d+1)(1-\beta\opt{\kappa})}\Big]^{1/d}\Big\}. \]
By Constraint~\ref{Margconstr:yStar}, it holds that $\hat{y}=\bar{n}/2$, and the corresponding 
linear program condition is covered in Constraint~\ref{Margconstr:upos_lessthann}.

\noindent --Consider the scenario where $1\leq x\leq \bar{n}/2-1$ and $y\geq\bar{n}/2+1$, where we require that
\[ (1-\beta\opt{\kappa})y^{d+1}
    -\gamma^\infty x^{d+1}
    +\opt{F}(x) x
    -\opt{F}(x+1) y
    -\opt{\kappa} \sum^{\bar{n}/2}_{k=x+1} \opt{F}(k)
    +\beta \opt{\kappa} \Big(\frac{\bar{n}}{2}\Big)^{d+1} \geq 0. \]
Observe that the left-hand side is convex since $1-\beta\opt{\kappa}>0$ by the condition in \eqref{Margconstr:bklessthan1}
and is minimized for nonnegative real values $y\geq\bar{n}/2+1$ at
\[ \hat{y}=\max\Big\{\frac{\bar{n}}{2}+1,\Big[\frac{\opt{F}(x+1)}{(d+1)(1-\beta\opt{\kappa})}\Big]^{1/d}\Big\}. \]
Observe that the resulting linear program conditions are satisfied 
for $y=\hat{y}$ and for all $1\leq x\leq \bar{n}/2$ since $\gamma^\infty \leq \gamma_1$,
for $\gamma_1$ as defined in \eqref{Margeq:gamma1}.

\noindent --We now consider the setting where $x\geq \bar{n}/2$ and $0\leq y<\bar{n}/2$.
Here we require:
\begin{align*}
    \gamma^\infty x^{d+1} & \leq y^{d+1}
    + \beta[x^{d+1}-(x-1)^{d+1}]x - \beta[(x+1)^{d+1}-x^{d+1}]y \\
    & \hspace*{100pt} + \opt{\kappa} \sum^{\bar{n}/2-1}_{k=y+1}\opt{F}(k) + \beta\opt{\kappa}\Big(x^{d+1}-\Big(\frac{\bar{n}}{2}-1\Big)^{d+1}\Big).
\end{align*}
Observe that the resulting linear program constraints are satisfied for all $0\leq y<\bar{n}/2$ 
since $\gamma^\infty \leq \gamma_2$ as defined in \eqref{Margeq:gamma2}.

\noindent --For $x,y\geq \bar{n}/2$, we require 
\[ \gamma^\infty x^{d+1}\leq y^{d+1} 
    + \beta[x^{d+1}-(x-1)^{d+1}]x - \beta[(x+1)^{d+1}-x^{d+1}]y
    + \beta\opt{\kappa} (x^{d+1}-y^{d+1}). \]
The left-hand side is convex in $y$ as $1-\beta\opt{\kappa}>0$ and is minimized 
over the nonnegative reals $y\geq \bar{n}/2$ by
\[ \hat{y} = \max\Big\{\frac{\bar{n}}{2}, \Big[\frac{\beta[(x+1)^{d+1}-x^{d+1}]}{(d+1)(1-\beta\opt{\kappa})}\Big]^{\frac{1}{d}}\Big\}. \]
The resulting linear program constraints are satisfied as $\gamma^\infty \leq \gamma_3$
as defined in \eqref{Margeq:gamma3}.

Observe that in the above, we have not explicitly used the fact that $1\leq x+y\leq n$.
Thus, as we did for the upper bound on the Price of Anarchy, here we prove that the constraints 
with $x+y > n$ are less strict that the constraints with $1\leq x+y\leq n$ for $(\opt{gamma},1,\opt{\kappa})$.
In fact, this amounts to showing once more that
\[ \nu xF^\infty(x) - \nu yF^\infty(x+1) \leq \nu (n-y)F^\infty(x) - \nu (n-x)F^\infty(x+1) \iff \nu (x+y-n)[F(x)-F(x+1)] \leq 0. \]
This must hold, since, $\nu=1>0$, $x+y-n>0$ and $F^\infty$ is nondecreasing.

\section{Computing a lower bound on the attainable joint performance guarantees} \label{appendix:gamebygame_lowerbound}

Given the maximum number of users $n$, consider a game parameterization corresponding with a $(8^n)\times(3n)$ table $\mc{R}$ whose rows are all the unique permutations of $3n$-long binary vectors $\mathbf{r}\in\{0,1\}^{3n}$.  Under this parameterization, each row $\mathbf{r}\in\mc{R}$ corresponds with a different resource, and, collectively, the rows encode the users' actions $a^\text{w-ne}_i, a^\text{b-ne}_i, a^\text{opt}_i$ as follows: 

Consider the resource $e$ corresponding with the row $\mathbf{r}\in\mc{R}$. For any $i\in\{1,\dots,n\}$, if $\mathbf{r}_i=1$, then $e \in a^\text{w-ne}_i$, else $e \notin a^\text{w-ne}_i$; if $\mathbf{r}_{n+i}=1$, then $e \in a^\text{b-ne}_i$, else $e \notin a^\text{b-ne}_i$; and, if $\mathbf{r}_{2n+i}=1$, , then $e \in a^\text{opt}_i$, else $e \notin a^\text{opt}_i$.  The coefficients in the basis representation of the $8^n$ resource cost functions will be the decision variables of our final linear program (i.e., there are $8^n\times m$ decision variables).

Recall from Section \ref{sec:gamebygame_lowerbound} that a lower bound on the Price of Stability of $\mc{G}^\tau_T$ can be computed as
\begin{equation} \label{eq:gamebygame_lowerbound_unparam}
\begin{aligned}
    \underset{G\in\mc{G}^\tau_T}{\text{maximize}} \quad &{\rm SC}(a^\text{b-ne}) \quad \text{subject to:} \\
    &{\rm SC}(a^\text{opt}) = 1, \quad {\rm SC}(a^\text{w-ne}) = \tau, \\
    &C_i(a^\text{w-ne})-C_i(a'_i,a^\text{w-ne}_{-i}) \leq 0, \quad \forall a'_i \in \{a^\text{b-ne}_i, a^\text{opt}_i\}, \quad \forall i\in\{1,\dots,n\}, \\
    &C_i(a^\text{b-ne}_{1:i},a'_{i+1:n})-C_i(a^\text{b-ne}_{1:i-1},a'_{i:n}) < 0, \quad\forall a' \in \mc{A}\setminus \{a^\text{w-ne}\}, \quad \forall i\in\{1,\dots,n\}.
\end{aligned}
\end{equation}
Recast under the above parameterization, we observe that the problem of computing the lower bound in \eqref{eq:gamebygame_lowerbound_unparam} can be formulated as the following linear program:
\begin{equation} \label{eq:lowerbound_lp_gamebygame}
\begin{aligned}
    \underset{\alpha_{e,j} \geq 0}{\text{maximize}} \quad &\sum_{e\in\mc{E}} \sum^m_{j=1} \alpha_{e,j} \cdot |a^\text{b-ne}|_e b_j(|a^\text{b-ne}|_e) \quad \text{subject to:} \\
    &\sum_{e\in\mc{E}} \sum^m_{j=1} \alpha_{e,j} \cdot |a^\text{opt}|_e b_j(|a^\text{opt}|_e) = 1, \quad
    \sum_{e\in\mc{E}} \sum^m_{j=1} \alpha_{e,j} \cdot |a^\text{w-ne}|_e b_j(|a^\text{w-ne}|_e) = \tau, \\
    &\sum_{e \in a^\text{w-ne}_i \setminus a'_i} \sum^m_{j=1} \Big[ \alpha_{e,j} \cdot F_j(|a^\text{w-ne}|_e) \Big] - \sum_{e \in a'_i \setminus a^\text{w-ne}_i} \sum^m_{j=1} \Big[ \alpha_{e,j} \cdot F_j(|a^\text{w-ne}|_e+1) \Big] \leq 0, \\
    &\hspace{200pt}\forall a'_i \in \{a^\text{b-ne}_i, a^\text{opt}_i\}, \quad \forall i\in\{1,\dots,n\}, \\
    &\sum_{e\in a^\text{b-ne}_i \setminus a'_i} \sum^m_{j=1} \Big[ \alpha_{e,j} \cdot F_j(|a^\text{b-ne}_{1:i},a'_{i+1:n}|_e) \Big] - \sum_{e\in a'_i \setminus a^\text{b-ne}_i} \sum^m_{j=1} \Big[ \alpha_{e,j} \cdot F_j(|a^\text{b-ne}_{1:i-1},a'_{i:n}|_e) \Big] < 0, \\
    &\hspace{200pt}\forall a' \in \mc{A}\setminus \{a^\text{w-ne}\}, \quad \forall i\in\{1,\dots,n\},
\end{aligned}
\end{equation}
where we use $\mc{E}$ to denote the set of resources corresponding with the rows of $\mc{R}$, and $a'_i\setminus a_i = a'_i \setminus (a'_i\cap a_i)$ denotes the resources that user $i$ selects in $a'_i$ and not in $a_i$, e.g., if $e\in a^\text{w-ne}_i\setminus a^\text{b-ne}_i$, then the corresponding binary vector $\mathbf{r}$ must have $\mathbf{r}_i=1$ and $\mathbf{r}_{n+i}=0$.  For given joint action $a$, observe that the quantity $|a|_e$ is simply the number of 1's in the appropriate columns of the binary vector $\mathbf{r}$ corresponding with $e$, i.e., $|a|_e = \sum^n_{i=1} \mathbf{r}_{nj_i+i}$ where $j_i = 0$ if $a_i=a^\text{w-ne}_i$, $j_i = 1$ if $a_i=a^\text{b-ne}_i$ and $j_i = 2$ if $a_i=a^\text{opt}_i$.  Note that the objective and constraints in \eqref{eq:gamebygame_lowerbound_unparam} and \eqref{eq:lowerbound_lp_gamebygame} coincide for games with a maximum number of users $n$.

As mentioned, the number of decision variables in \eqref{eq:lowerbound_lp_gamebygame} grows as $8^n\times m$ in the maximum number of users $n$ and number of basis functions $m$, and, thus, the linear program is not tractable for large $n$.  Nonetheless, lower bounds for small $n$ can be computed within a reasonable amount of time, as provided in Figure \ref{fig:gamebygame}.

\end{appendix}

\end{document}